\documentclass[a4paper,11pt]{article}
\pdfoutput=1

\usepackage{amsmath,amssymb,amsfonts,mathrsfs}
\usepackage{dsfont}
\usepackage{tikz}
\usepackage{tikz-cd}
\usepackage{todonotes}
\usetikzlibrary{cd}
\usepackage[all]{xy}
\usetikzlibrary{shapes,arrows}
\usepackage{mathtools}
\usepackage{amsthm}

\usepackage{jheppub}

\def\bp{\begin{pmatrix}}
\def\ep{\end{pmatrix}}

\usepackage{graphicx} 
\usepackage[utf8]{inputenc}
\usepackage{url}
\usepackage{hyperref}
\usepackage{amsmath,amssymb,tikz-cd}
\usepackage{tikz-cd}
\usepackage{braket}
\usepackage{amsthm}
\usepackage{xcolor}
\newtheorem{theorem}{Theorem}
\definecolor{red}{rgb}{1,0,0}
\long\def\comment#1{}
\usepackage{chessfss}
\pdfoutput=1
\pdfpageattr {/Group << /S /Transparency /I true /CS /DeviceRGB>>}
\usepackage{hyperref}
\hypersetup{colorlinks=true,linkcolor=blue,citecolor=blue}
\usepackage{amsmath,amssymb,amsfonts}
\usepackage{bm}
\usepackage{float}
\usepackage{subfig}
\usepackage{scalefnt}
\usepackage{wrapfig}
\usepackage{todonotes}
\usepackage{fancybox}
\usepackage{graphicx}
\usepackage{tikz}
\usepackage{multirow}
\usepackage{adjustbox}
\usepackage[toc,page]{appendix}

\title{\boldmath Irregular KZ equations and Kac-Moody representations}
\author[a]{Sergei Gukov, }
\author[b,c]{Babak Haghighat, }
\author[b]{Yihua Liu, }
\author[b,c,d]{Nicolai Reshetikhin}

\affiliation[a]{Richard N. Merkin Center for Pure and Applied Mathematics, California Institute of Technology, Pasadena, CA 91125, USA}
\affiliation[b]{Yau Mathematical Sciences Center, Tsinghua University, Beijing, 100084, China}
\affiliation[c]{Beijing Institute of Mathematical Sciences and Applications (BIMSA), Huairou District, Beijing 101408, China}
\affiliation[d]{St. Peresburg University, 7-9 Universitetskaya Embankment, St Petersburg, 199034, Russia}

\abstract{In this paper we construct irregular representations of the affine Kac-Moody algebra $\widehat{sl}(2,\mathbb{C})$. We show how such irregular representations correspond to irregular Gaiotto-Teschner representations of the Virasoro algebra. The intertwiners for such representations satisfy a version of Knizhnik-Zamolodchikov (KZ) equations which we call irregular KZ equations. By connecting to 2d Liouville theory, we show how the conformal blocks governed by our irregular KZ equation correspond to 4d Argyres-Douglas theories with surface operator insertions. The corresponding flat connections describe braiding between such operators on the Gaiotto curve.}

\begin{document}

\maketitle

\section{Introduction}

Correlation functions of two-dimensional conformal field theories (CFTs) are well-understood and widely used when the corresponding operators are primary fields in highest weight representations. The corresponding representations can be those of the Virasoro algebra, W-algebras or Kac-Moody algebras in the case of WZW models. In many cases, one can utilize symmetries to derive differential equations satisfied by conformal blocks with operators in such regular representations. These equations can be Belavin-Polyakov-Zamolodchikov (BPZ) equations \cite{Belavin:1984vu} in the presence of degenerate fields or Knizhnik-Zamolodchikov (KZ) equations \cite{KNIZHNIK198483} in the presence of affine Kac-Moody algebras. KZ equations have moreover the remarkable property that they define a flat connection on the conformal block bundle. This allows one to construct 3d knot invariants by using these flat connections to transport operators around each other adiabatically along a third direction \cite{Witten:1988hf}. 

Operators in so-called irregular representations of 2d CFTs are less studied but have recently gained more importance in the context of the AGT or BPS/CFT correspondence \cite{Alday:2009aq,Gaiotto:2009ma,Dijkgraaf:2009pc,Gaiotto:2011nm,Nagoya,Gaiotto:2012sf,Nekrasov:2015wsu}. Before worrying how to define such representations, one can see their effect in the differential equations satisfied by the corresponding conformal blocks. Namely, the resulting differential equations must have irregular singular points and the corresponding KZ equations were initially introduced in \cite{Resh-KZ}. On the 4d gauge theory side of the AGT correspondence, these singularities map to higher order poles of the Higgs field associated to the Seiberg-Witten geometry of the supersymmetric gauge theory \cite{Gaiotto:2009hg}. Thus, one obtains flat connections with irregular singularities. Such theories are known as Argyres-Douglas theories \cite{Argyres:1995jj} and have been studied extensively using Painlave equations in the past \cite{Bonelli:2016qwg}. This approach uses the theory of isomonodromic deformations of systems of linear ODEs or flat connections in our terminology. We also refer to the related works \cite{Jimbo_2008,Awata:2010bz,Nagoya_2010,Nishinaka:2012kn,Rim:2012tf,Gaiotto:2013rk,Choi:2014qha,Nagoya:2015cja,Ashok:2016yxz,Bonelli:2022ten,Felder:2020ugk} which consider integral representations and differential equations in the presence of irregular operators from different perspectives. \cite{Awata:2009ur,Awata:2010bz} performed initial studied of BPZ type differential equations satisfied by irregular conformal blocks. These results were later used in \cite{Ashok:2016yxz,Bonelli:2022ten} to derive monodromy transformations and corresponding Stokes matrices. \cite{Nishinaka:2012kn,Rim:2012tf,Choi:2014qha,Nagoya:2015cja,Frenkel:2015rda,Nekrasov:2017gzb} consider integral representations of Liouville conformal blocks with degenerate and irregular operators, but they do not connect to current algebras and WZW models. On the other hand, \cite{Jimbo_2008,Nagoya_2010,Felder:2020ugk} do consider KZ equations and integral representations of $Sl(2)$ WZW conformal blocks with irregular operators, but don't connect these to Liouville conformal blocks. This connection is made later on in \cite{Gaiotto:2013rk} based on the general relation found in \cite{Ribault:2005wp}. Results from a more mathematical perspective can be found in \cite{JIMBO1981306,Feigin_2010a,Feigin_2010b,XuQG}.  However, the representation theoretic relation between irregular Virasoro conformal blocks and their Kac-Moody counterparts remains elusive.\footnote{In \cite{Ribault:2005wp,Gaiotto:2013rk} the relationship is given in terms of a complicated variable change but a concrete relationship between the two integral representations is lacking.}

The goal of the present paper is to remedy this gap by starting a program to derive irregular KZ equations from irregular Virasoro conformal blocks using a representation theoretic framework. This will naturally give rise to irregular flat connections for conformal block bundles. A first step in this direction was made in \cite{Haghighat:2023vzu}. In the current paper, our starting point are Kac-Moody representations where we will be interested in the affine $Sl(2,\mathbb{C})$ case. This will be relevant for the Liouville CFT while higher rank current algebras will be relevant for various Toda CFTs. The strategy employed is to use a bosonic Fock space representation for the modes of the algebra such that positive modes up to a certain degree do not annihilate the vacuum. This leads to induced modules which are irregular. We then show via the Sugawara construction that, under certain restrictions, such irregular modules give rise to Gaiotto-Teschner type representations \cite{Gaiotto:2012sf} of the Virasoro algebra relevant for Argyres-Douglas theories \cite{Argyres:1995jj}. Using intertwining operators, we then derive KZ equations for conformal blocks with an irregular vacuum at infinity and n highest weight representations at specified points on $\mathbb{P}^1$. Within the framework of Liouville theory, such regular representations correspond to the insertion of degenerate fields and become surface operators in the dual 4d gauge theory \cite{Alday:2009fs}. Finally, we derive integral solutions of our irregular KZ equations in cases where the irregular singularity is restricted to be of degree one. We show how these integral solutions can be derived from a master function by taking suitable derivatives. 

\subsection{Summary of results}
The main results of the paper are as follows. In Section \ref{sec:irrkcm}, we define irregular Kac-Moody representations by modifying the action of the modes. The main result here is that for both $gl(1)$ and $sl(2)$ we can show that Gaiotto-Teschner representations \cite{Gaiotto:2012sf} of the Virasoro algebra can be obtained via the Sugawara construction. What is more, we can obtain generalizations by constructing so-called \textit{fully irregular} $sl(2)$ modules of degree $m$, depending on parameters $\{x_1,\ldots,x_m,y_1,\ldots,y_m,z_1,\ldots ,z_m\}$. The Sugawara construction then gives the following action of Virasoro generators on such states:
\begin{align}
    \hat{L}_0 &=L_0+\frac{k}{k+2}\sum_{i=1}^m i(\frac{\partial}{\partial y_i}y_i-\frac{\partial}{\partial z_i}z_i)),\\
    \hat L_i &= L_i-\frac{k}{k+2}(\sum_{j=1}^i jz_jy_{i-j}+z_0\sum_{j=1}^i z_j+\sum_{j=1}^{m-i}(j\frac{\partial}{\partial y_j}y_{i+j}-(i+j)\frac{\partial}{\partial z_j}z_{i+j})), \quad \text{for $i$=1,2,...m}\\
    \hat L_i &= L_i-\frac{k}{k+2}(\sum_{j=i-m}^{m} jz_{j}y_{i-j}), \quad \text{for $i=m+1,...2m$}\\  L_n &=0,\quad \text{for $n>2m$}.
\end{align}
In the above, the $L_i$ denote the original irregular Virasoro representation constructed in \cite{Gaiotto:2012sf} where the parameters $c_l$ with $l=1, \ldots, m$, appearing in \cite{Gaiotto:2012sf}, are identified with our parameters $x_l$ (see section \ref{sec:irrkcm} for further details).

In Section \ref{sec:KZeq} and the rest of the paper, we focus on irregular Kac-Moody representations with $y_i=z_i=0$ which give rise to the Gaiotto-Teschner state constructed in \cite{Gaiotto:2012sf,Nagoya}. By utilizing the correspondence between Liouville conformal blocks and $SL(2,\mathbb{C})$ WZW blocks (see \cite{Cordova:2016cmu} for a modern viewpoint), we show how irregular versions of KZ equations at level $\kappa = - b^2$ can be obtained by placing such an irregular state at infinity. For degree 1 irregular singularities at infinity we arrive at the following modified KZ equation
\begin{equation} \label{eq:introKZ}
      -b^2\frac{\partial}{\partial z_i} \psi = \frac\Lambda2 H_i \psi + \left(\sum_{j \neq i} \frac{\Omega_{ij}}{z_i-z_j}\right)\psi,
\end{equation}
for the $sl(2,\mathbb{C})$ case, where $H_i$ is the Cartan generator acting on the $i$th regular representation. Performing the variable change $z_i \mapsto \frac{1}{z_i}$, one can see that the above equation develops a double pole when $z_i \rightarrow \infty$ and is hence irregular of degree one there. Irregular KZ equations of the form \eqref{eq:introKZ} and certain generalizations therefore were studied in \cite{Jimbo_2008}. What makes our approach different is the somewhat different derivation and the connection to Liouville theory as established in Section \ref{sec:heights} and Theorem \ref{th:height2}. 

We then proceed to show that integral solutions of this equation can be obtained via Liouville conformal blocks with $N$ degenerate field insertions and one degree-1 irregular operator at infinity. Concretely, our main result is the following. Define $N$ point degree-1 irregular Liouville CFT conformal block with $m$ screening operators, with $k_i \in \mathbb{Z}$, $b \in \mathbb{C}$, $\Lambda \in \mathbb{R}$, as follows:
    \begin{align}
        \mathcal{F}(z)=\exp\left(-\frac{\Lambda}{2b^2}\sum_{i=1}^N {k_i z_i}\right)\prod_{i<j}(z_i-z_j)^{-\frac{k_ik_j}{2b^2}}\int A({\bold{z,w}}) d^m w=\phi_0(z) \phi(\mathbf{z}) ,
    \end{align}
    where we define
    \begin{align}
        \phi(\mathbf{z}) \equiv \int A({\bold{z,w}})= \int \exp\left(\frac{\Lambda}{b^2}\sum_i w_i\right)\prod_{i<j}(w_i-w_j)^{-\frac{2}{b^2}}\prod_{i,j}(w_i-z_j)^\frac{k_j}{b^2}.
    \end{align}
    For integer $\bold{m}=(m_1,m_2,...,m_N)$, 
  $m_i\geq 0$, $\sum_i m_i=m$, let
    \begin{align}
       \phi^{(\bold{m})}= \int A({\bold{z,w}}) d^m w \left(\prod_{i=1}^{m_1}\frac{1}{w_{i_1}-z_1}...\prod_{i=1}^{m_N}\frac{1}{w_{i_N}-z_N}\right).
    \end{align}
The $\phi^{(\bold{m})}$ can be obtained via the action of $m$th order differential operators on the master function $\phi(\mathbf{z})$ which we explicitly show for the cases $m=1$ and $m=2$ (see Theorem \ref{th:height2}), and conjecture to hold for arbitrary $m$.
Then, the column vector $\widehat\psi$ containing all $\phi^{(\bold{m})}$ will satisfy the following KZ equation (Theorem \ref{th:heightm})\footnote{The KZ-equations \eqref{eq:introKZ} and \eqref{eq:TmKZ} are equivalent upon making the identification $\psi = \phi_0 \widehat\psi$.}
\begin{align}\label{eq:TmKZ}
   -b^2 \partial_i  \widehat\psi=A_i\widehat\psi+\sum_{j\neq i}\frac{\widehat{\Omega}^{T}_{ij}}{z_i-z_j}\widehat\psi,
\end{align}
where $\widehat{\Omega}^{T}$ is the transpose of the $\widehat{\Omega}$ matrix of height  $m$ KZ equations, see Section \ref{sec:heights} for the definition and concrete matrix representations. And the action of $A_i$ is 
\begin{align}
   A_i \phi^{(\bold{m})}=\Lambda m_i  \phi^{(\bold{m})}.
\end{align}

\section{Irregular Kac-Moody representations}
\label{sec:irrkcm}

In this section, we will derive irregular representations of the affine Kac-Moody algebra. We first start with the $\widehat{gl(1)}$ case. Using the Sugawara construction to construct the stress-energy tensor, we find that its action on the irregular affine $\widehat{gl(1)}$ state reproduces an irregular Virasoro state as defined in \cite{Gaiotto:2012sf}. We then generalize the above discussion to $\widehat{sl(2)}$ by using the Wakimoto construction, and show that an irregular Virasoro state can again be reproduced.

\subsection{The $\widehat{gl(1)}$ case}
In this section we want to construct irregular affine $\widehat{gl(1)}$ representations. The strategy will be to use a Fock space representation of the Heisenberg algebra $\hat{\alpha}$ of bosonic modes $\alpha_n$, $n \in \mathbb{Z}$, with commutation relations
\begin{equation}
    [\alpha_n, \alpha_m] = n \delta_{n+m}.
\end{equation}
We can use these modes to form the current
\begin{equation}
    \alpha(z) \equiv \sum_{n \in \mathbb{Z}} \alpha_n z^{-n-1}.
\end{equation}
Using the Sugawara construction of the stress energy tensor, we will see how the irregular Kac-Moody states we construct can be identified with irregular Gaiotto-Teschner states \cite{Gaiotto:2012sf} of the Virasoro algebra.

\paragraph{The irregular degree $r=1$ construction.}
To begin with, we recall the construction of the regular affine $\widehat{gl(1)}$ highest weight representations. Here, the Fock space is given by the following polynomial ring,
\begin{align}
    \mathcal{H}=\textit{Pol}(x_1,x_2,....).
\end{align}
The action of the generators $\left\{ \alpha_i , i\in\mathbb{  Z} \right\}  $ of  the affine Lie algebra $\widehat{gl(1)}$ is then given by
\begin{align}
    &\alpha_0\mathcal{H}=\lambda \mathcal{H} ,\\
    &\alpha_{n}\mathcal{H}=n\frac{\partial}{\partial x_n}\mathcal{H},\textit{for\space} n>0,\\
    &\alpha_{-n}\mathcal{H}=x_n\mathcal{H},\textit{for\space} n>0.
\end{align}
The vacuum or highest weight state is given by $\Omega=\mathbb{C}.1$. This realizes the usual Fock module $\mathcal{F}=\textit{Ind}_{\hat{\alpha}_{\leq 0}}^{\hat{\alpha}}\Omega$, where $\alpha_0 \Omega=\lambda \Omega, \alpha_{n} \Omega=0 $ for $n>0$.

Let us now turn to the irregular Fock space representation of $\hat{\alpha}$. To this end, we modify the action of the modes $\alpha_n$ as follows:
\begin{align}
     &\alpha_0\mathcal{H}=\lambda \mathcal{H} ,\\
    &\alpha_{n}\mathcal{H}=n\frac{\partial}{\partial x_n}\mathcal{H},\textit{for\space} n>1, \\
&\alpha_{-n}\mathcal{H}=x_n\mathcal{H},\textit{for\space} n>1,\\
&\alpha_{-1}\mathcal{H}=\frac{\partial}{\partial x_1}\mathcal{H},\\
&\alpha_{1}\mathcal{H}=-x_1\mathcal{H}.
\end{align}
One can check that this also forms a representation of $\hat{\alpha}$. Now in this case, we can construct the corresponding stress-energy tensor using the Sugawara construction. To proceed, we define $\alpha(z) = \sum_{n\in \mathbb{Z}}\alpha_n z^{-n-1}=\alpha_{-}(z)+\alpha_{+}(z)$, where
\begin{align}
   &\alpha_{-}(z)=\sum_{n<0}\alpha_n z^{-n-1}+\alpha_0 z^{-1},\\
   &\alpha_{+}(z)=\sum_{n>0}\alpha_n z^{-n-1}.
\end{align}

Using the following commutator,
\begin{align}
   \left[\alpha_{+}(z),\alpha_{-}(w)\right]&=\left[\sum_{n>0}\alpha_n z^{-n-1},\sum_{m<0}\alpha_m w^{-m-1}\right] \nonumber\\
   &=\sum_{n>0,m<-0}n\delta_{n+m} z^{-n-1} w^{-m-1}\nonumber\\
   &=\sum_{n>0}n z^{-n-1} w^{n-1}~,
\end{align}
we see that
\begin{align}\label{eq:free}
   \left[\alpha_{+}(z),\alpha_{-}(w)\right] &=\frac{1}{zw}\left.\frac{d}{dx}\left(\sum_{n\geq0}x^{n+1}-\sum_{n\geq 0}x^n\right)\right|_{x=\frac{w}{z}}\nonumber\\
   &=\frac{1}{zw}\left.\frac{x}{(1-x)^2}\right|_{x=\frac{w}{z}}\nonumber\\
   &=\frac{1}{(z-w)^2},
\end{align}
and
\begin{align}
    &\left[\alpha_{+}(z),\alpha_{+}(w)\right]=0 ,\\
    &\left[\alpha_{-}(z),\alpha_{-}(w)\right]=0.
\end{align}
Thus the corresponding 2-point function $\braket{\alpha(z)\alpha(w)}$ in our case is just \eqref{eq:free}, and the constraint $\left|z\right|>\left|w\right|$ is just the radial quantization condition. Next, we define the stress-energy tensor as follows
\begin{align} \label{eq:Tgl1}
T(z)=\frac{1}{2}:\alpha(z)\alpha(z):+\theta \partial \alpha(z),
\end{align} 
where the term proportional to $\theta$ indicates a curvature insertion at infinity.
In order to check that the modes of the stress-energy tensor satisfy the Virasoro algebra, it is enough to check the operator product expansion of $T$ with itself:
\begin{align}\label{eq:Valgebra}
    T(z)T(w)=\frac{\frac{c}{2}}{(z-w)^4}+\frac{2 T(w)}{(z-w)^2}+\frac{1}{z-w}\partial T(w).
\end{align} 
Using Wick's theorem, we get 
\begin{align}
    \mathcal{R}(T(z)T(w))&=\frac{1}{4} (2\braket{\alpha(z)\alpha(w)}^2+4\braket{\alpha(z)\alpha(w)}\alpha(z)\alpha(w))+\frac{\theta}{2}2\braket{\partial\alpha(z)\alpha(w)}\alpha(w)\nonumber\\
    &+\frac{\theta}{2}2\braket{\alpha(z)\partial\alpha(w)}\alpha(z)+\theta^2\braket{\partial\alpha(z)\partial\alpha(w)}...\nonumber\\
    &=\frac{1}{2}\frac{1}{(z-w)^4}+\frac{1}{(z-w)^2}\alpha(z)\alpha(w)+\frac{2\theta}{(z-w)^3}(\alpha(z)-\alpha(w))\nonumber\\
    &-\frac{6\theta^2}{(z-w)^4}\nonumber\\
    &=\frac{1}{2}\frac{1-12\theta^2}{(z-w)^4}+\frac{1}{(z-w)^2}(\alpha(w)\alpha(w)+2\theta\partial\alpha(w))\nonumber\\
    &+\frac{1}{(z-w)}(\partial(\frac{1}{2}\alpha(w)\alpha(w))+\partial(\theta \partial \alpha(w))+...\nonumber\\
    &=\frac{1}{2}\frac{c}{(z-w)^4}+\frac{2 T(w)}{(z-w)^2}+\frac{1}{z-w}\partial T(w),
\end{align}
which from \eqref{eq:Valgebra} fixes the central charge to the value $c=1-12\theta^2$. For the modes of the stress-energy tensor one obtains
\begin{align}
    L_n&= \frac{1}{2}\sum_{n\in \mathbb{Z}} :\alpha_m \alpha_{n-m}:+\theta(-n-1)\alpha_n ,\quad \text{for $n\neq0$} \\
    L_0&=\sum_{m\in \mathbb{N^{+}}}\alpha_{-m}\alpha_{m}+\frac{1}{2}\alpha_0^2-\theta\alpha_0~.
\end{align}
Next, we need to choose our irregular vacuum $\Omega$, represented in terms of a specific function $f(x_1)$. Choosing $f=\frac{1}{\sqrt{x_1}}$ ensures that $-\frac{\partial}{\partial x_1}(x_1 f)=x_1\frac{\partial}{\partial x_1}$, which in turn implies that the action of Virasoro generators reads as
\begin{align}
    L_0\Omega&=(x_1\frac{\partial}{\partial x_1}+\frac{\lambda^2}{2}-\theta \lambda)f(x_1),\\
    L_1\Omega&=-x_1(\lambda -2\theta)f(x_1),\\
    L_2\Omega&=\frac{1}{2}x_1^2 f(x_1).\\
    L_n\Omega&=0,\text{for $n>2$}.
\end{align}
The above reproduces the definition of irregular Gaiotto-Teschner states of the Virasoro algebra \cite{Gaiotto:2012sf}. To see this, perform the following substitutions,
\begin{align*}
    x_1&\rightarrow \sqrt{2}\sqrt{-1}x_1, \\
    \theta&\rightarrow-\frac{\sqrt{2}}{2}\sqrt{-1}\theta,\\
    \lambda&\rightarrow- \sqrt{2}\sqrt{-1}\lambda,
\end{align*}
giving
\begin{align}
    L_0\Omega&=(x_1\frac{\partial}{\partial x_1}+\lambda(\theta-\lambda))f(x_1),\\
    L_1\Omega&=-2x_1(\lambda -\theta)f(x_1),\\
    L_2\Omega&=-x_1^2 f(x_1).
\end{align}
If we identify $\theta=Q,\lambda=\alpha,x_1=c$, then the above will exactly match with the degree 1 irregular state as defined in \cite{Gaiotto:2012sf}.  

\paragraph{Generalization to higher degree $r > 1$.}
Furthermore, if we want to generalize this procedure to recover the degree 2 or higher irregular singular state, we just need to consider the following isomorphic but different representation of the Affine $gl(1)$ algebra. Take degree 2 for example, we may consider the following:
\begin{align}
    &\alpha_0\mathcal{H}=\lambda \mathcal{H} ,\\
    &\alpha_{-1}\mathcal{H}=\frac{\partial}{\partial x_1}\mathcal{H},\\
    &\alpha_{1}\mathcal{H}=-x_1\mathcal{H},\\
    &\alpha_{-2}\mathcal{H}=2\frac{\partial}{\partial x_2}\mathcal{H},\\
    &\alpha_{2}\mathcal{H}=-x_2\mathcal{H},\\
&\alpha_{n}\mathcal{H}=n \frac{\partial}{\partial x_n}\mathcal{H},\textit{for\space} n>2, \\
    &\alpha_{-n}\mathcal{H}=x_n\mathcal{H},\textit{for\space} n>2.
\end{align}
In this case, the form of the stress-energy tensor is kept invariant since the above definition doesn't change the structure of $gl(1)$ affine Kac-Moody algebra. But now the vacuum state $\Omega$ should be modified as 
\begin{align}
    \Omega=h=f_1(x_1)f_2(x_2).
\end{align}
And now the exact action of the Virasoro modes reads:
\begin{align}
     L_0\Omega&=-2\frac{\partial}{\partial x_2}(x
_2 h)-\frac{\partial}{\partial x_1}(x
_1 h)+(\frac{\lambda^2}{2}-\theta \lambda)h,\\
    L_1\Omega&=x_1(\lambda -2\theta)h-\frac{\partial}{\partial x_1 }(x_2 h),\\
    L_2\Omega&=(\frac{1}{2}x_1^2+x_2\lambda-3\theta\lambda)h ,\\
    L_3\Omega&=x_2x_1 h,\\
    L_4\Omega&=\frac{1}{2}x_2^2 h,\\ 
    L_n\Omega&=0,\text{for $n>4$}.
\end{align}
Performing the substitutions:
\begin{align*}
    x_1&\rightarrow \sqrt{2}\sqrt{-1}x_1, \\
    x_2&\rightarrow \sqrt{2}\sqrt{-1}x_2, \\
    \theta&\rightarrow-\frac{\sqrt{2}}{2}\sqrt{-1}\theta,\\
    \lambda&\rightarrow -\sqrt{2}\sqrt{-1}\lambda,
\end{align*}
the above becomes:
\begin{align}
      L_0\Omega&=-2\frac{\partial}{\partial x_2}(x
_2 h)-\frac{\partial}{\partial x_1}(x
_1 h)+\lambda(\theta-\lambda)h,\\
    L_1\Omega&=-2x_1(\lambda -\theta)h-x_2\frac{\partial}{\partial x_1 }h,\\
    L_2\Omega&=(-x_1^2-x_2(2\lambda-3\theta))h ,\\
    L_3\Omega&=-2x_2x_1 h,\\
    L_4\Omega&=-x_2^2 h,\\ 
    L_n\Omega&=0,\text{for $n>4$}.
\end{align}
We also need to specify to the state $h=\sqrt{x_1 x_2}$ giving rise to the following action
\begin{align}
    -\frac{\partial}{\partial x_2}(x
_2 h)=x_2 \frac{\partial}{\partial x_2}h,-\frac{\partial}{\partial x_1}(x
_1 h)=x_1 \frac{\partial}{\partial x_1}h.
\end{align}
Finally, the action of the Virasoro generators becomes:
\begin{align}
      L_0\Omega&=(2x_2\frac{\partial}{\partial x_2}+x_1\frac{\partial}{\partial x_1}+\lambda(\theta-\lambda))h,\\
    L_1\Omega&=-2x_1(\lambda -\theta)h-x_2\frac{\partial}{\partial x_1 }h,\\
    L_2\Omega&=(-x_1^2-x_2(2\lambda-3\theta))h ,\\
    L_3\Omega&=-2x_2x_1 h,\\
    L_4\Omega&=-x_2^2 h,\\ 
    L_n\Omega&=0,\text{for $n>4$}.
\end{align}
One can see that this is the degree 2 irregular state as defined in \cite{Gaiotto:2012sf}. One can also check that the commutation relation
\begin{align}
    \left[L_2,L_1\right]=L_3
\end{align}
is satisfied. So indeed we obtain a faithful representation of the Virasoro algebra.

It's easy to construct the degree $r$, or the rank $r$ (using the wording in \cite{Gaiotto:2012sf}), irregular state by employing the following definition:
\begin{align}
    &\alpha_0\mathcal{H}=\lambda \mathcal{H} ,\\
    &\alpha_{-i}\mathcal{H}=i\frac{\partial}{\partial x_i}\mathcal{H},\textit{for\space} r\geq i\geq 1\\
    &\alpha_{i}\mathcal{H}=-x_i\mathcal{H},\textit{for\space} r\geq i\geq 1\\
&\alpha_{n}\mathcal{H}=n \frac{\partial}{\partial x_n}\mathcal{H},\textit{for\space} n>r, \\
    &\alpha_{-n}\mathcal{H}=x_n\mathcal{H},\textit{for\space} n>r,
\end{align}
 setting the vacuum state to $\Omega_r=\sqrt{\prod_{i=1}^r x_i}$, and performing the substitutions:
     \begin{align*}
    x_i&\rightarrow \sqrt{2}\sqrt{-1}x_i,\text{for $r\geq i\geq 1$}, \\
    \theta&\rightarrow-\frac{\sqrt{2}}{2}\sqrt{-1}\theta,\\
    \lambda&\rightarrow -\sqrt{2}\sqrt{-1}\lambda.
\end{align*}
After these steps, one arrives at the following action of Virasoro modes:
\begin{align}
    L_0\Omega_r&=(\lambda(\theta-\lambda)+\sum_{m=1}^r mx_m\frac{\partial}{\partial x_m})\Omega_r,\\
    L_{k}\Omega_r&=-x_{k}(2\lambda-(k+1)\theta)-\sum_{m+l=k} x_m x_{l}-\sum_{j=1}^{r-k}j x_{k+j} \frac{\partial}{\partial x_j},\text{for $r-1\geq k\geq 1$},\\
    L_k\Omega_r&=-\sum_{m+l=k} x_m x_{l}, \text{for $k\in\left\{ r+1,...2r\right\}$},\\
    L_{k}\Omega_r&=0,\text{for $k>2r$.}
\end{align}
This is nothing but the irregular state of degree $r$ defined in \cite{Gaiotto:2012sf}.
\subsection{The $\widehat{sl(2)}$ case}

Now the next step is to generalize the above construction to the $\widehat{sl(2)}$ case. Here we should consider the Wakimoto construction, or the ``bosonization" procedure, namely the 2 copies of the Heisenberg algebra. And the generators of $sl(2)$ are constructed from two copies of the Heisenberg algebra, with the following definition:
\begin{align} 
    e(z)&=\beta(z),\nonumber\\
    h(z)&=-2:\gamma(z)\beta(z):+\sqrt{2\kappa}\,\alpha(z),\nonumber\\
    f(z)&=-:\gamma^2(z) \beta(z):+ \sqrt{2\kappa}\,\alpha(z)\gamma(z)+k\frac{d\gamma(z)}{dz}. \label{eq:sl2currents}
\end{align}
Here $\kappa=k+2$, and $k$ is the level of the Kac-Moody algebra. Furthermore, $e,f,h,\beta,\gamma,\alpha$ can be written in terms of series expansions of modes as follows
\begin{align}
    &e(z)=\sum_{n\in \mathbb{Z}}e_n z^{-n-1},~f(z)=\sum_{n\in \mathbb{Z}}f_n z^{-n-1},~h(z)=\sum_{n\in \mathbb{Z}}h_n z^{-n-1},\\
    &\beta(z)=\sum_{n\in \mathbb{Z}}\beta_n z^{-n-1},~\alpha(z)=\sum_{n\in \mathbb{Z}}\alpha_n z^{-n-1},\\
    &\gamma(z)=\sum_{n\in \mathbb{Z}}\gamma_n z^{-n}.
\end{align} 
$\beta,\gamma,\alpha$ are generators of Heisenberg algebras with respect to the following commutation relations:
\begin{align}
    &\left[\beta_n,\gamma_m \right]
    =\delta_{n+m},\\
    &\left[\alpha_n,\alpha_m \right]    =n\delta_{n+m}, 
\end{align}
and all other commutators vanish. The normal ordering  of the  $\beta\gamma$ system is defined to be 
\begin{align}
    :\beta_m\gamma_n:=\begin{cases}
\gamma_n \beta_m & \text{ if } m\geq0 \\
\beta_m\gamma_n & \text{ if } m<0 
\end{cases}.
\end{align}
The stress energy tensor is now given again in terms of the Sugawara construction as follows,
\begin{align}
    T(z)=\frac{1}{2(k+2)}(:I_1^2(z):+:I_2^2(z):+:I_3^2(z):),
\end{align}
where $I_i$ are an orthonormal basis of $sl(2,\mathbb{C})$ with respect to the Killing form($K(X,Y)=\text{Tr}(XY)$). For our case we pick 
\begin{align}
    I_1=\frac{h}{\sqrt{2}},I_2=e+\frac{1}{2}f,I_3=ie-\frac{i}{2}f,
\end{align} 
Then 
\begin{align}
    T(z)&=\frac{1}{2(k+2)}(:I_1^2(z):+:I_2^2(z):+:I_3^2(z):)\nonumber\\
    &=\frac{1}{2(k+2)}(\frac{1}{2}:h^2(z):+:e(z)f(z):+:f(z)e(z):),
\end{align}
where
\begin{align}
    &:h^2(z):=4:\gamma^2(z)\beta^2(z):+2(k+2):\alpha^2(z):-4\sqrt{2(k+2)}*\alpha(z):\beta(z)\gamma(z):\\
    &:e(z)f(z):=-:\gamma^2(z)\beta^2(z):+\sqrt{2(k+2)}*\alpha(z):\beta(z)\gamma(z): +k:\frac{d\gamma(z)}{dz}\beta(z):\\
    &:f(z)e(z):= :e(z)f(z):~.
\end{align}
Thus
\begin{align}
    T(z)&=\frac{1}{2(k+2)}\left(\frac{2(k+2)}{2}:\alpha^2(z):+2k:\frac{d\gamma(z)}{dz}\beta(z):\right)\nonumber\\
    &=\frac{1}{2}:\alpha^2(z):+\frac{k}{k+2}:\frac{d\gamma(z)}{dz}\beta(z):\nonumber\\
    &=\frac{1}{2}\left(\sum_{m,n\in \mathbb{Z}}:\alpha_m\alpha_n: z^{-m-n-2}\right)+\frac{k}{k+2}\left(\sum_{m,n\in \mathbb{Z}}-n:\gamma_n\beta_m: z^{-m-n-2}\right).
\end{align}
Furthermore, we can add the term $\theta \partial\alpha(z)$ to the stress energy tensor as in the $gl(1)$ case while preserving the Virasoro algebra since $\alpha(z)$ commutes with $\frac{d\gamma(z)}{dz}\beta(z)$.

We then consider an irregular action of $\alpha$ modes while retaining a regular action for the $\beta$ and $\gamma$ modes:
\begin{align}
     &\alpha_0\mathcal{H}=\lambda \mathcal{H} ,\\
    &\alpha_{n}\mathcal{H}=n\frac{\partial}{\partial x_n}\mathcal{H},\textit{for\space} n>1, \\
&\alpha_{-n}\mathcal{H}=x_n\mathcal{H},\textit{for\space} n>1,\\
&\alpha_{-1}\mathcal{H}=\frac{\partial}{\partial x_1}\mathcal{H},\\
&\alpha_{1}\mathcal{H}=-x_1\mathcal{H}.
\end{align}

\begin{align}
    &\beta_{n}\mathcal{H}=\frac{\partial}{\partial y_n}\mathcal{H},\textit{for\space} n\geq1, \\
&\beta_{-n}\mathcal{H}=z_n\mathcal{H},\textit{for\space} n\geq0.
\end{align}

\begin{align}
    &\gamma_{n}\mathcal{H}=-\frac{\partial}{\partial z_n}\mathcal{H},\textit{for\space} n\geq0, \\
&\gamma_{-n}\mathcal{H}=y_n\mathcal{H},\textit{for\space} n\geq1.
\end{align}
Then for $\Omega_1=f(x_1)=\frac{1}{\sqrt{x_1}}$ we have:
\begin{align}
    L_0\Omega_1&=(x_1\frac{\partial}{\partial x_1}+\frac{\lambda^2}{2}-\theta \lambda)f(x_1),\\
    L_1\Omega_1&=-x_1(\lambda -2\theta)f(x_1),\\
    L_2\Omega_1&=\frac{1}{2}x_1^2 f(x_1),\\
    L_n\Omega_1&=0,\text{for $n>2$}.
\end{align}
Which can be simplified by scaling to obtain the  rank $1$ Gaiotto-Teschner state as in \cite{Gaiotto:2012sf}. 
We thus conclude that the correct form of the  have stress-energy tensor is
\begin{align}\label{sl2stress}
    T(z)=\frac{1}{2}:\alpha^2(z):+\frac{k}{k+2}:\frac{d\gamma(z)}{dz}\beta(z):-\frac{i\theta}{\sqrt{2}}\partial\alpha(z),
\end{align}
with the action of the $\alpha$ modes specified as
\begin{align}
     &\alpha_0\mathcal{H}=-i\sqrt{2}\lambda\mathcal{H} ,\\
    &\alpha_{n}\mathcal{H}=n\frac{\partial}{\partial x_n}\mathcal{H},\textit{for\space} n>1, \\
&\alpha_{-n}\mathcal{H}=x_n\mathcal{H},\textit{for\space} n>1,\\
&\alpha_{-1}\mathcal{H}=-i\frac{1}{\sqrt{2}}\frac{\partial}{\partial x_1}\mathcal{H},\\
&\alpha_{1}\mathcal{H}=-i\sqrt{2}x_1\mathcal{H}.
\end{align}
Furthermore, define the currents
\begin{align}
    h_-(z) &= \sum_{n > 0} h[n] z^{-n-1} + h[0] z^{-1},  \\
    h_+(z) &= \sum_{n<0} h[n] z^{-n-1},
\end{align}
and define similarly $e_{\pm}(z),~f_{\pm}(z)$. After the above rescaling,  we have the following action of such currents on the left vector $\Omega_1$ for the standard inner product of the vector space:
\begin{align}
     \langle  \Omega_1,h_+(z_i)  \cdots\rangle&=\langle \sum_{n>0}h_n z^{n-1} \Omega_1,  \cdots\rangle = \langle h_1  \Omega_1,  \cdots\rangle =-2i\sqrt{\kappa}x_1 \langle h_1  \Omega_1,  \cdots\rangle,\\
    \langle  \Omega_1,e_+(z_i)  \cdots\rangle&=\langle \sum_{n>0}e_n z^{n-1} \Omega_1,  \cdots\rangle=\langle \sum_{n>0}\beta_n z^{n-1} \Omega_1,  \cdots\rangle = 0,\\
    \langle  \Omega_1,f_+(z_i)  \cdots\rangle&=\langle \sum_{n>0}f_n z^{n-1} \Omega_1,  \cdots\rangle \propto\langle \sum_{n>0}(\sum_{l+m=n}\alpha_l\gamma_ m) z^{n-1} \Omega_1,  \cdots\rangle\nonumber\\
    &=\langle \alpha_1\gamma_0 z^{n-1} \Omega_1,  \cdots\rangle=0.
\end{align}

Also one can get any degree $r$ irregular state as in the previous section. In this case the action of the currents on the irregular state $\Omega_r$ changes accordingly to
\begin{align}
    \langle  \Omega_r,h_+(z_i)  \cdots\rangle&=\langle \sum_{n>0}h_n z^{n-1} \Omega_r,  \cdots\rangle= \langle \sum_{n > 0}^r \sqrt{2\kappa} \alpha_n z^{n-1} u_0, \cdots \rangle \nonumber \\ &=-2i\sqrt{\kappa}\sum_{n > 0}^r x
_n z^{n-1}\langle   \Omega_r,  \cdots\rangle,\\
 \langle  \Omega_r,e_+(z_i)  \cdots\rangle&= 0,\\
 \langle  \Omega_r,f_+(z_i)  \cdots\rangle&=\langle \sqrt{2\kappa} \sum_{n+m>0} \alpha_n \gamma_m z^{n+m-1} \Omega_r,\cdots \rangle \nonumber \\
    &=\langle \sqrt{2\kappa} \sum_{n=1}^r\sum_{m=1-n}^{-1} \alpha_n \gamma_m z^{n+m-1} \Omega_r,\cdots \rangle\nonumber\\
    &=-\sqrt{2\kappa}\langle  \sum_{n=1}^r\sum_{m=1-n}^{-1} x_n y_m z^{n+m-1} \Omega_r,\cdots \rangle.
\end{align}
The above relations are necessary for deriving the KZ equation in the next section.

\subsection{Beyond Gaiotto-Teschner state}
In general, we can create full rank 1 irregular representations for the $\widehat{sl(2)}$ Kac-Moody algebra by letting the action of $\beta,\gamma$ be irregular as well, given by:
\begin{align}
    &\beta_{n}\mathcal{H}=\frac{\partial}{\partial y_n}\mathcal{H},\textit{for\space} n>1, \\
&\beta_{-n}\mathcal{H}=z_n\mathcal{H},\textit{for\space} n\geq0,n \neq 1.\\
&\beta_{1}\mathcal{H}=y_1\mathcal{H}.\\
&\beta_{-1}\mathcal{H}=\frac{\partial}{\partial z_1}\mathcal{H}.\\
    &\gamma_{n}\mathcal{H}=-\frac{\partial}{\partial z_n}\mathcal{H},\textit{for\space} n\geq0, n \neq 1 \\
&\gamma_{-n}\mathcal{H}=y_n\mathcal{H},\textit{for\space} n>1,\\
&\gamma_1\mathcal{H}=z_1\mathcal{H},\\
&\gamma_{-1}\mathcal{H}=-\frac{\partial}{\partial y_1}\mathcal{H}.
\end{align}
Then the action of Virasoro generators can be written as:
\begin{align}
    L_n=\frac{1}{2}\sum_{m,l\in \mathbb{Z},m+l=n} :\alpha_m\alpha_l:+\frac{k}{k+2}\sum_{m,l\in \mathbb{Z},m+l=n} -m:\gamma_m \beta_l:+i\frac{\theta}{\sqrt{2}} (n+1)\alpha_n
\end{align}
For the vacuum state $\Omega=\frac
{1}{\sqrt{x_1}}f(y_1,z_1)$ one then obtains:
\begin{align}
    L_0\Omega&=(x_1\frac{\partial}{\partial x_1}+\lambda(\theta-\lambda)+\frac{k}{k+2}(\frac{\partial}{\partial y_1}y_1-\frac{\partial}{\partial z_1}z_1))f,\\
    L_1\Omega&=(-2x_1(\lambda -\theta)-\frac{k}{k+2}z_1z_0)f(x_1),\\
    L_2\Omega&=(-x_1^2+-\frac{k}{k+2}y_1z_1) f(x_1),\\
     L_n\Omega&=0,\text{for $n>2$}.
\end{align}
This construction can be generalized to the full rank $m$ irregular state by the following:
\begin{align}
    &\beta_{n}\mathcal{H}=\frac{\partial}{\partial y_n}\mathcal{H},\textit{for\space} n>m, \\
&\beta_{-n}\mathcal{H}=z_n\mathcal{H},\textit{for\space $n=0$  or $n>m$,} \\
&\beta_{i}\mathcal{H}=y_i\mathcal{H},i=1,2,...,m\\
&\beta_{-i}\mathcal{H}=\frac{\partial}{\partial z_i}\mathcal{H},i=1,2,...,m\\
    &\gamma_{n}\mathcal{H}=-\frac{\partial}{\partial z_n}\mathcal{H},\textit{for\space $n=0$  or $n>m$,}  \\
&\gamma_{-n}\mathcal{H}=y_n\mathcal{H},\textit{for\space} n>m,\\
&\gamma_i\mathcal{H}=z_i\mathcal{H},i=1,2,...,m\\
&\gamma_{-i}\mathcal{H}=-\frac{\partial}{\partial y_i}\mathcal{H},i=1,2,...,m.
\end{align}
And the action of Virasoro generators $\hat{L}_n$ on the corresponding irregular state \begin{equation}
    \Omega=\frac
{1}{\sqrt{x_1x_2..x_m}}f(\{y_i\},\{z_i\}),
\end{equation} now reads:
\begin{align}
    \hat{L}_0 &=L_0+\frac{k}{k+2}\sum_{i=1}^m i(\frac{\partial}{\partial y_i}y_i-\frac{\partial}{\partial z_i}z_i)),\\
    \hat L_i &= L_i-\frac{k}{k+2}(\sum_{j=1}^i jz_jy_{i-j}+z_0\sum_{j=1}^i z_j+\sum_{j=1}^{m-i}(j\frac{\partial}{\partial y_j}y_{i+j}-(i+j)\frac{\partial}{\partial z_j}z_{i+j})), \quad \text{for $i$=1,2,...m}\\
    \hat L_i &= L_i-\frac{k}{k+2}(\sum_{j=i-m}^{m} jz_{j}y_{i-j}), \quad \text{for $i=m,m+1,...2m$}\\  L_n &=0,\quad \text{for $n>2m$},
\end{align}
where the $L_i$ furnish the original irregular Gaiotto-Teschner representation.
\section{KZ equations}
\label{sec:KZeq}

In this section we will derive KZ equations for conformal blocks with one irregular operator at infinity and $N$ regular operators at position $z_i$. To this end, we will first introduce an intertwining operator for irregular states and then use such operators as building blocks to form conformal blocks. The commutation relations of the intertwining operator with current operator modes, known as gauge invariance, will then naturally lead to a KZ equation for a given current algebra. 

\subsection{Intertwining operators}
The intertwining operator is defined to be the operator $\Phi:V_{f_1}\rightarrow V_{f_2}\otimes V(z)$ such that they satisfy the intertwining condition :
\begin{align} \label{eq:gaugeinv1}
    \Phi a[n]=(a[n]\otimes 1+1\otimes z^n x)\Phi,
\end{align}
where $V(z)$ is an evaluation representation from a finite representation $V$ of a Lie algebra $\frak{g}$, and $a \in \frak{g}$. In order to obtain the KZ equation, $V(z)$ will later be chosen as the tensor product of the highest weight representations of $\widehat{\frak{g}}$. By definition, we have for elements $v \in V(z)$: 
\begin{align}
 a[n] v= a z^n v, \quad \text{for $v\in V$}.
\end{align}
The modules $V_{f_i}$ correspond to regular or irregular representations induced by vacuum states $f_i$. The relation \eqref{eq:gaugeinv1} in known as \textit{gauge invariance}. Another way to state this constraint is to define for $u \in V^*$ and $v \in V_{f_1}$ the function $\Phi_u(z)v \equiv u(\Phi(z)v)$. Then one can regard $\Phi_u(z)$ as an operator $V_{f_1} \rightarrow V_{f_2}$. The intertwining relation \eqref{eq:gaugeinv1} then translates to
\begin{equation} \label{eq:gaugeinv2}
    [a[n],\Phi_u(z)] = z^n \Phi_{a u}(z).
\end{equation}
Defining the currents
\begin{align}
    a_-(z) &= \sum_{n > 0} a[n] z^{-n-1} + a[0] z^{-1},  \\
    a_+(z) &= \sum_{n<0} a[n] z^{-n-1},
\end{align}
one can show that \eqref{eq:gaugeinv2} implies
\begin{equation} \label{eq:gaugeinv3}
    [a_{\pm}(w), \Phi_u(z)] = \frac{1}{z-w} \Phi_{a u}(z).
\end{equation}
Let us in the following apply this construction to the irregular $\widehat{gl(1)}$ and $\widehat{sl(2)}$ representations studied previously.

\subsection{The $\widehat{gl(1)}$ case}
We consider an irreducible $gl(1)$ representation $V$ which is one-dimensional. Thus $V(z)$ is also one-dimensional and we have the map $V_{f_1} \rightarrow V_{f_2}\otimes V(z)$ given by the intertwiner $\Phi_{u=1}$. Moreover, the expression \eqref{eq:Tgl1} for the stress-energy tensor gives us
\begin{align} \label{eq:opKZ}
    z \frac{d}{dz} \Phi_1(z) &= [L_0,\Phi_1(z)] - \Delta \Phi_1(z) = \frac{1}{2} z : \alpha(z) \Phi_{\alpha_0}(z) :,
\end{align}
Actually, if we use the modified version of the stress-energy tensor as in \eqref{eq:Tgl1}, then we will have an extra term $\alpha_0\theta\Phi_1(z)=\lambda \theta\Phi_1(z)$ on the right hand side of the equation. But this can be compensated by shifting the weight of the representation with the extra term $\lambda\theta$ because now the form of $L_0$ has changed for the extra term of $\theta\alpha_0$, but luckily acting on the highest weight representation it is still a constant. So the equation remains the same. The same argument holds for the $\widehat{sl(2)}$ case.

The normal ordering gives
\begin{equation}
    : \alpha(z) \Phi_{\alpha_0}(z) : = \alpha_+(z) \Phi_{\alpha_0}(z) + \Phi_{\alpha_0}(z)\alpha_-(z).
\end{equation}
We will now choose $V_{f_2}$ to be an irregular module given by $f_2 = \frac{1}{\sqrt{x_1}}$. Let $u_0$ be the irregular vacuum of the module $\left(V_{\frac{1}{\sqrt{x_1}}}\right)^*$ and $u_{N} \in V_1[0] = \mathbb{C} \cdot 1$ be the degree zero state of the regular module. Then, consider the \textit{conformal block} 
\begin{equation}
    \psi(z_1,\ldots,z_{N-1}) = \langle u_0, \Phi_1(z_1)\ldots \Phi_1(z_{N-1})u_N\rangle \in \mathbb{C}.
\end{equation}
Now we compute
\begin{align}
    \frac{\partial}{\partial z_i} \psi(z_1,\ldots,z_{N-1}) &= 2 \langle u_0, \Phi_1(z_1) \ldots \frac{\partial}{\partial z_i} \Phi_1(z_i)\ldots \Phi_1(z_{N-1})u_{N}\rangle \nonumber \\
    &= \langle u_0, \Phi_1(z_1)\ldots (a_+(z)\Phi_{\alpha_0}(z_i) + \Phi_{\alpha_0}(z_i)a_-(z))\ldots \Phi_1(z_N)u_N\rangle.
\end{align}
Next, we can use the commutation relations \eqref{eq:gaugeinv3} to pull the currents all the way to the right and left:
\begin{align}
    ~ &~ \langle u_0, \Phi_1(z_1)\ldots (\alpha_+(z)\Phi_{\alpha_0}(z_i) + \Phi_{\alpha_0}(z_i)\alpha_-(z))\ldots \Phi_1(z_N)u_N\rangle \nonumber \\
    ~ &= \sum_{j=1}^{i-1}\langle u_0, \Phi_1(z_1) \ldots [\Phi_1(z_j),\alpha_+(z_i)]\ldots \Phi_{\alpha_0}(z_i)\ldots \Phi_1(z_{N-1})u_N\rangle \nonumber \\
    ~ &\quad + \sum_{j=i+1}^{N-1}\langle u_0, \Phi_1(z_1)\ldots \Phi_{\alpha_0}(z_i)\ldots [\Phi_1(z_j),\alpha_-(z)]\ldots \Phi_1(z_{N-1})u_N\rangle \nonumber \\
    ~ &\quad + \langle u_0, \alpha_+(z_i)\Phi_1(z_1)\ldots \Phi_{\alpha_0}(z_i) \ldots \Phi_1(z_{N-1})u_N\rangle \nonumber \\
    ~ &\quad + \langle u_0, \Phi_1(z_1)\ldots \Phi_{\alpha_0}(z_i) \ldots \Phi_1(z_{N-1})\alpha_-(z_i)u_N\rangle~.
\end{align}
To obtain the final expression for the KZ equation, we just have to note that $\alpha_-(z)u_N = \alpha_0 u_n/z$, and the irregular action
\begin{equation}
    \langle u_0, \alpha_+(z) \cdots \rangle = - \langle \alpha_+(z) u_0, \cdots\rangle = x_1 \langle u_0, \cdots \rangle.
\end{equation}
Here the action is reversed because under the contragredient representation,  the action of $\alpha_n$ on $u_0$ gives $\alpha_{-n}$. 
Upon identifying $x_1 \equiv \Lambda$, this then gives the final form of our KZ equation
\begin{equation}
    2 \frac{\partial}{\partial z_i} \psi = \Lambda \alpha_0 \psi + \left(\sum_{j \neq i} \frac{\alpha_0^2}{z_i-z_j}\right)\psi.
\end{equation}

\subsection{The $\widehat{sl(2)}$ case}
The $sl(2)$ case is now analogous to the $gl(1)$ case with the only difference being that the operator KZ equation \eqref{eq:opKZ} is changed to
\begin{equation}\label{intertwiner}
    z \frac{d}{dz} \Phi_u(z) = \frac{1}{k + 2} z \left(\frac{1}{2} :h(z) \Phi_{H u}(z): + :e(z) \Phi_{Y u}(z): + : f(z) \Phi_{X u}(z): \right),
\end{equation}
where $h(z)$, $e(z)$ and $f(z)$ are the $\widehat{sl(2)}$ currents. \eqref{eq:sl2currents}.
\paragraph{Degree $r=1$.}
In order to derive the corresponding KZ equation, one has to follow the same steps as above. The irregular piece of the expression is given by
\begin{align}
   ~ &~ \frac{1}{2} \langle u_0, h_+(z_i) \Phi_u(z_1) \ldots \Phi_{h u}(z_i) \ldots \Phi_u(z_{N-1})u_N\rangle \nonumber \\
   ~ &+ \langle u_0, f_+(z_i) \Phi_u(z_1) \ldots \Phi_{e u}(z_i) \ldots \Phi_u(z_{N-1})u_N\rangle \nonumber \\
   ~&+ \langle u_0, e_+(z_i) \Phi_u(z_1) \ldots \Phi_{f u}(z_i) \ldots \Phi_u(z_{N-1})u_N\rangle.
\end{align}
The equation then follows from
\begin{equation}
    \langle h_+(z_i) u_0,  \cdots\rangle \propto \langle u_0, \cdots \rangle, \quad \langle e_+(z_i) u_0,  \cdots\rangle = 0, \quad \langle f_+(z_i) u_0,  \cdots\rangle = 0.
\end{equation}
Because in this case:
\begin{equation}
    \langle h_+(z_i) u_0,  \cdots\rangle=-2i\sqrt{\kappa}x_1\langle   u_0,  \cdots\rangle.
\end{equation} $-2i\sqrt{\kappa}x_1=-2i(ib)x_1=2bx_1$.
\begin{equation}
    \langle e_+(z_i) u_0,  \cdots\rangle=0,
\end{equation}
\begin{equation}
    \langle f_+(z_i) u_0,  \cdots\rangle=0.
\end{equation}
So we obtain the following irregular KZ equation for degree $1$ irregular Kac-Moody algebra with the location of irregular state at infinity:
\begin{equation}
      \kappa\frac{\partial}{\partial z_i} \psi = b x_1 H_i \psi + \left(\sum_{j \neq i} \frac{\Omega_{ij}}{z_i-z_j}\right)\psi,
\end{equation}
where $\Omega_{ij}$ is defined as
\begin{equation}
    \Omega_{ij} \equiv \frac{1}{2} H_i H_j + X_i Y_j + Y_i X_j,
\end{equation}
with 
\begin{equation}
    H \equiv h[0], \quad X \equiv e[0], \quad Y \equiv f[0].
\end{equation}

\paragraph{Degree $r=2$.}
For the degree $r=2$ case, we compute 
\begin{align}
    \langle h_+(z_i) u_0, \cdots\rangle = \langle \sqrt{2\kappa}(-x_1 - x_2 z) u_0, \cdots \rangle.
\end{align}
The action of $e_+(z)$ on the irregular vacuum $u_0$ still vanishes as $\beta(z)$ acts regularly. However, now we also have a non-trivial action of $f_+(z)$ on $u_0$ that was absent in the degree one case. To see this, note that the only terms contributing non-trivially are those containing $\alpha$ modes and thus we get
\begin{align}
    \langle f_+(z_i) u_0, \cdots \rangle &=  \langle \sqrt{2\kappa} (-x_2 y_1)u_0,\cdots \rangle.
\end{align}
Since $f_+(z_i) u_0$ is paired with the action of $X_i$ on the $i$th state, we therefore obtain the following irregular KZ equation in the degree $2$ case:
\begin{equation} \label{eq:r2KZ}
    \kappa \frac{\partial}{\partial z_i} \psi = \left(\frac{1}{2}\sqrt{2\kappa}(-x_1 - x_2 z_i) H_i - \sqrt{2\kappa} x_2 y_1 X_i + \sum_{j \neq i} \frac{\Omega_{ij}}{z_i - z_j} \right) \psi. 
\end{equation}

\subsection{Relation to irregular Liouville conformal blocks: Heights}\label{sec:heights}

Liouville conformal field theory is a non-unitary CFT with infinitely many primaries, where all primary operators have a vertex operator representation given by $V_\lambda(z)$ with $\lambda$ corresponding to the weight of the primary field. Moreover, conformal blocks of degenerate operators can be computed in the free field or so-called \textit{Coulomb gas} formalism. This implies that conformal blocks admit an integral representation corresponding to the free bosonic CFT partition function via integration over suitable contours. In the case where there is an irregular vertex operator, the free field construction also exists \cite{Dijkgraaf:2009pc,Gaiotto:2011nm,Gaiotto:2012sf,Nishinaka:2012kn,Rim:2012tf,Haghighat:2023vzu,Gu:2023plq}, but the integration contour may well be non-compact. More specifically, one finds for conformal blocks $\mathcal{F}$ of $N$ degenerate vertex operators and one degree one irregular operator at infinity corresponding to a countour $\Gamma$ the following result:
\begin{equation} \label{eq:Firr}
    \mathcal{F}_{\Gamma} = \braket{\prod_{a=1}^N V_{\frac{-k_a}{2b}}(z_a)I(\infty)} = \int_{\Gamma} \exp\left(\frac{1}{b^2} \mathcal{W}_{\mathrm{irr}}\right)\prod_i dw_i, \quad k_a \in \mathbb{Z},~b \in \mathbb{C},
\end{equation}
where $\mathcal{W}_{\mathrm{irr}}$ is given by
\begin{eqnarray}
    \mathcal{W}_{\mathrm{irr}} & \equiv & \sum_{i<j} -2 \log(w_i-w_j) + \sum_{i,a} k_a \log(w_i-z_a) - \sum_{a<b} \frac{1}{2} k_a k_b \log(z_a - z_b) \nonumber \\
    ~ & ~ & + \Lambda\left(\sum_i w_i - \frac{1}{2} \sum_a k_a z_a \right). \label{eq:Wirr}
\end{eqnarray}
For future convenience, it will be useful to abbreviate the integral representation as follows
\begin{equation}
    \mathcal{F}_{\Gamma} = \phi_0(z) \phi_{\Gamma}(z),
\end{equation}
where
\begin{align}
    \phi_0 &= \exp\left(-\frac{\Lambda}{2b^2} \sum_{a=1}^N k_a z_a\right)\prod_{a<b} (z_a - z_b)^{-\frac{k_a k_b}{2b^2}}, \\
    \phi_{\Gamma} &= \int_{\Gamma} \exp\left(\frac{\Lambda}{b^2}\sum_{i=1}^m w_i\right) \prod_{i < j} (w_i - w_j)^{-\frac{2}{b^2}}\prod_{i,a} (w_i - z_a)^{\frac{k_a}{b^2}} d^m w. \label{eq:phi}
\end{align}

We will see later that there is a correspondence between solutions of irregular KZ equations and the irregular conformal blocks of Liouville CFT. That is, particular combinations of derivatives of irregular Liouville conformal blocks will satisfy irregular KZ equations!  The number of screening operators in the Coulomb gas representation of irregular Liouville conformal blocks then corresponds to the height of a subspace in the tensor product of $\widehat{sl(2)}$ representations. By height we mean here the number of roots lower than the highest weight in the tensor product. What's more, the charges of the vertex operators in Liouville CFT will correspond to the weight of 
the $\widehat{sl(2)}$ representation.  Finally, the differential equation satisfied by derivatives of the irregular conformal blocks will be exactly the same as the differential equation that we derived in the $\widehat{sl(2)}$ case as a generalization of \cite{Haghighat:2023vzu}, if the following relationship between the two is imposed:
\begin{align}\label{LtoKZ}
    c=x_1=\frac{\Lambda}{2b},b=-i\kappa~.
\end{align}
This correspondence is expected from the general relation between Liouville theory and $SL(2,\mathbb{C})$ WZW theory, see for example \cite{Cordova:2016cmu}.
After the substitution \eqref{LtoKZ}, we can see that the degree-1 KZ equation will have the following form:
\begin{equation}\label{height1kz}
      -b^2\frac{\partial}{\partial z_i} \psi = \frac\Lambda2 H_i \psi + \left(\sum_{j \neq i} \frac{\Omega_{ij}}{z_i-z_j}\right)\psi,
\end{equation}
which will be later shown to hold when $\psi$ is a vector of (higher) derivatives of irregular Liouville conformal blocks.

One can also see this correspondence by comparing the number of conformal blocks in both cases. If we restrict our discussion to $N$-point conformal blocks of $1/2$ spin representations of $\widehat{sl(2)}$, then from the fusion rules one can see that there are $2^N$  conformal blocks, since there will be in total $N$ nodes in the fusion trees, and each tree can only have labels $0$ or $1$ from the fusion rules of affine $\widehat{sl(2)}$ representations.

\begin{equation} \label{eq:cfblocktree}
    \langle \phi_{\lambda_1}(z_1)\ldots \phi_{\lambda_N}(z_n)I(\infty)\rangle_{\mu_1,\ldots,\mu_{N-2}}=\adjincludegraphics[valign=c]{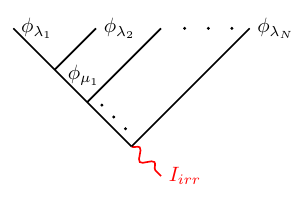}
\end{equation}

To make the connection to Liouville blocks, note that identifying $\frac{k_a}{2}$ with the highest weight $\lambda_a$, one obtains for spin $\frac{1}{2}$ operators $V_{\frac{-1}{2b}}$. Indeed, their fusion rule
\begin{equation}
    V_{-\frac{1}{2b}} \times V_{-\frac{1}{2b}} = V_0 + V_{-\frac{1}{b}},
\end{equation}
reflects the fusion rule of two spin one-half fields. Now, for the $N$-point irregular conformal block $\braket{\prod_{i=1}^N V_{\frac{-1}{2b}}(z_i)I(\infty)}$ with the irregular operator located at infinity, the number of conformal blocks is the number of Lefschetz thimbles, that is, the number of independent integration cycles. Such integration thimbles are composed of steepest descent paths flowing out from critical points of a superpotential $\mathcal{W}(w,z)$. So for a critical point $\sigma$ of $\mathcal{W}(w,z)$, the associated Lefschetz thimble is obtained by joining all paths $\mathcal{W}(t)$, $t \in (-\infty,0]$ which solve the first order equation
\begin{equation}
    \frac{d\bar{w}}{dt} = - \frac{\partial \mathcal{W}}{\partial w}, \quad \frac{dw}{dt} = - \frac{\partial \bar{\mathcal{W}}}{d \bar{w}}, \quad \mathcal{W}(-\infty) = \sigma. 
\end{equation}
Such thimbles can either arise from flows between two critical points but in more general situations are non-compact and flow to infinity along both ends. For $\braket{\prod_{i=1}^N V_{\frac{-1}{2b}}(z_i)I(\infty)}$, we have
\begin{eqnarray}
    \mathcal{W} & \equiv & \sum_{i<j} -2 \log(w_i-w_j) + \sum_{i,a}  \log(w_i-z_a) - \sum_{a<b} \frac{1}{2} \log(z_a - z_b) \nonumber \\
    ~ & ~ & + \Lambda\left(\sum_i w_i - \frac{1}{2} \sum_a  z_a \right). \label{eq:Wirr}
\end{eqnarray}
One can see that the number of critical points of $\mathcal{W}$ depends on the number of screening charges, or the number of $w$ variables. For $i$ screening charges, there are $\binom{N}{i}$ independent paths. Because the number of paths can be determined from the equation of critical points of $\mathcal{W}$ :
\begin{align}
    \sum_j \frac{1}{w_i-z_j}=-2\sum_{j\neq i}\frac{1}{w_i-w_j}+\Lambda.
\end{align}
when we consider the limit $\Lambda\rightarrow\infty$, we can see that the $w_i$'s should be near one of the $z_j$'s. So this gives $\binom{N}{i}$ solutions for each of such choice. And in total we will have $\sum_{i=0}^n \binom{N}{i}=2^{N}$ paths, which matches the result from the representations of affine Kac-Moody algebra. And for general $k_i$ and $2$ $z_i$'s, this relation is shown in \cite{Hu2016}. We here conjecture that this relationship is also correct for general $k_i$ and any number of $z_i$'s.

Next we derive the exact KZ equation in different subspaces in terms of their heights, namely the subspace generated by the highest weight vector with the total number of lowering operators acting on it fixed. The form of the KZ equation will be block diagonal in different heights since the operator $\Omega$ preserves the heights.

\paragraph{Height 1.}
Here we want to derive the KZ equation obtained in \cite{Haghighat:2023vzu} from previous sections. To this end, note that the corresponding conformal blocks are valued in 
\begin{equation}
    V_{k_1} \otimes \cdots V_{k_N},
\end{equation}
where 
\begin{equation}
    V_k = \{ Y^l v_k, ~ l=0,1,2,\ldots |  X v_k = 0, ~H v_k = k v_k \},
\end{equation}
and 
\begin{equation}\label{hiwei}
    v_{k_1} \otimes v_{k_2} \otimes v_{k_N} \in V_{k_1} \otimes \cdots V_{k_N}.
\end{equation}
$H,X,Y$ satisfy :
\begin{align}
    \left[H,X\right]=2X, \left[H,Y\right]=-2Y,
     \left[X,Y\right]=H.
\end{align}
Now define
\begin{equation}
    v^{(a)} = v_{k_1} \otimes \cdots \otimes Y v_{k_a} \otimes \cdots \otimes v_{k_N},
\end{equation}
and set 
\begin{equation}
    \Omega_{ab} \equiv \frac{H_a H_b}{2} +  X_a Y_b + X_b Y_a \quad \textrm{where} \quad \alpha_a = 1 \otimes \cdots \otimes \underbrace{\alpha}_{a} \otimes \cdots \otimes 1,
\end{equation}

that is an operator $\alpha_a$ acts as $\alpha$ on the $a$th entry of the tensor product. Next, we compute
\begin{equation}
    \Omega_{cd} v^{(a)} = \frac{k_c k_d}{2} v^{(a)} \quad \textrm{if} \quad c,d \neq a, 
\end{equation}
and 
\begin{eqnarray}
    \Omega_{ab} v^{(a)} & = & \frac{H_a H_b}{2} v^{(a)} + Y_b X_a v^{(a)} \nonumber \\
    ~ & = & \frac{(k_a-2)k_b}{2} v^{(a)} + k_a v^{(b)} \nonumber \\
    ~ & = & \frac{k_a k_b}{2} v^{(a)} - k_b v^{(a)} + k_a v^{(b)}. 
\end{eqnarray}
Thus, we get that on weight one states of the form
\begin{equation}
    V^{(1)} \equiv \left(V_{k_1} \otimes \cdots \otimes V_{k_N}\right)^{(1)} = \bigoplus_{a=1}^N \mathbb{C}v^{(a)},
\end{equation}
the action of $\Omega_{cd}$ is given by
\begin{equation}
    \Omega_{cd} v^{(a)} = \frac{k_a k _d}{2} v^{(a)} + \widehat{\Omega}_{cd} v^{(a)}.
\end{equation}

So the action of $\Omega_{ad}$ has the following matrix expression:
\begin{align}
    \Omega_{ad}=\frac{k_a k_d}{2}I_N+\begin{pmatrix}
 -k_d& k_a \\
k_d & -k_a \\
\end{pmatrix}.
\end{align}
Thus, if we constrain ourselves to consider only ``height" one states, meaning including only one $Y$ operator acting on one of the highest weight vectors in $v_{k_1} \otimes \cdots \otimes  v_{k_a} \otimes \cdots \otimes v_{k_N}$, then we will get the following KZ equation:
\begin{align}
    \kappa\partial_i\psi=\left(\frac{\Lambda}{2}H_i+\sum_{j\neq i}\frac{\Omega_{ij}}{z_i-z_j}\right)\psi,
\end{align}
where $H_i$ has the following explicit representation,
\begin{align}
    H_i=\begin{pmatrix}
k_i &0  &0  &...  &0  \\
 0& ... &0  &...  & 0 \\
 0& ... &k_i-2  &...  &0  \\
 0& ... &0  &..  & 0 \\
 0& 0 &0  & ... &k_i  \\
\end{pmatrix}.
\end{align}

At this point we can connect to the results of \cite{Haghighat:2023vzu},  after we take into account the contribution of the free part, that is, the factor $(\exp(-\sum_i\frac{\Lambda k_i z_i}{2b^2}))(z_i-z_j)^{\frac{-k_i k_j}{2b^2}}$. Noting that, in the case of one screening charge, $\nabla \phi$ with $\phi$ given by \eqref{eq:phi} is a solution of the irregular KZ equation in \cite{Haghighat:2023vzu}, and letting 
\begin{align}
    \psi(z)=(\exp(\sum_i\frac{\Lambda k_i z_i}{2b^2}))\prod_{i>j}(z_i-z_j)^{\frac{-k_i k_j}{2b^2}}\nabla\phi(z)=f(z)\nabla\phi(z),
\end{align}
the equation satisfied by $\psi(z)$ is as follows:
\begin{align}
    \partial_i \psi&=(\partial_i f)\nabla\phi+f\partial_i(\nabla\phi)\nonumber\\
    &=(\partial_i f)\frac{\psi}{f}+f(A_i+\sum_{j\neq i}\frac{\widehat{\Omega}_{ij}}{b^2(z_i-z_j)})\nabla\phi \nonumber\\
    &=(-\sum_{j\neq i}\frac{k_i k_j}{2b^2(z_i-z_j)}-\frac{\Lambda k_i}{2b^2}) \psi+f(A_i+\sum_{j\neq i}\frac{\widehat{\Omega}_{ij}}{b^2(z_i-z_j)})\nabla\phi\nonumber\\
    &=-\frac{1}{b^2}(\frac{\Lambda}{2} H_i+\sum_{j\neq i}\frac{\Omega_{ij}}{z_i-z_j})\psi,
\end{align}
where $\widehat{\Omega}_{ij}$ is the Omega matrix specified in \cite{Haghighat:2023vzu}. 
Here we used that
\begin{align}
    \partial_i f=\left(-\sum_{j\neq i}\frac{k_i k_j}{2b^2(z_i-z_j)}-\frac{\Lambda k_i}{2b^2}\right) f.
\end{align}
Thus the column vector $\psi$ will satisfy the same equation as the height one KZ equation.

\paragraph{Height 2.}
Following the previous experience, we will figure out the case in the presence of two screening charges. Denote 
\begin{align}
    v^{(a,b)}&=v_{k_1} \otimes \cdots \otimes Y v_{k_a} \otimes \cdots \otimes Y v_{k_b}\otimes \cdots \otimes v_{k_n}, \\
    v^{(a,a)}&=v_{k_1} \otimes \cdots \otimes Y^2 v_{k_a} \otimes \cdots  \otimes v_{k_n}.
\end{align}
Similar to above, we have 
\begin{align}
    \Omega_{ab}v^{(c,d)}&=\frac{k_a k_b}{2} v^{(c,d)}, \label{abcd}\\ 
    \Omega_{ab}v^{(a,d)}&=\frac{(k_a-2) k_b}{2} v^{(a,d)}+k_a v^{(b,d)},\label{abad}\\ 
    \Omega_{ab}v^{(a,b)}&=\frac{(k_a-2) (k_b-2)}{2} v^{(a,b)}+k_a v^{(b,b)}+k_b v^{(a,a)},\label{abab}\\ 
    \Omega_{ab}v^{(a,a)}&=\frac{(k_a-4)k_b}{2} v^{(a,a)}+((k_a-2)+k_a) v^{(a,b)}\label{abaa}, 
\end{align}
by using the following:
\begin{align}
    XY^2-Y^2X=XY^2-YXY+YXY-Y^2X=(XY-YX)Y+Y(XY-YX)=HY+YH
\end{align}
Using these results, we can prove the following theorem regarding integral solutions of the height two KZ equation:
    \begin{theorem}\label{th:height2}
    (Integral expression for degree 1 Height 2 case) Consider the following degree one irregular $N$-point Liouville conformal blocks:
     \begin{align*}
   \phi&=\int_{\Gamma} dw_1 dw_2 A(w_1,w_2,z),\\
    A(w_1,w_2,z)&=\exp(\frac{\Lambda}{b^2}(w_1+w_2))(w_1-w_2)^{-\frac{2}{b^2}} \prod_{i=1}^n(w_1-z_i)^{\frac{k_i}{b^2}}\prod_{i=1}^n(w_2-z_i)^{\frac{k_i}{b^2}}.
\end{align*}
We further define
\begin{align}
    \phi^{(i,i)}&=k_i(k_i-1)\int A\frac{1}{(w_1-z_i)(w_2-z_i)},\\
    \phi^{(i,j)}&=\frac{k_i k_j}{2}\int A(\frac{1}{(w_1-z_i)(w_2-z_j)}+\frac{1}{(w_1-z_j)(w_2-z_i)}).
    \end{align}
    Then for $\psi=\begin{pmatrix}
\phi^{(1,1)} \\
 \phi^{(1,2)}\\
 ...\\
\phi^{(n,n)}
\end{pmatrix}$, we have\\
\\
    \noindent\textbf{1)}  $\psi$ will satisfy the degree 1 height 2 irregular KZ equation of  $\widehat{sl(2)}$ with $\kappa=-b^2$, 
    \begin{align}
        \partial_i\psi= A_i\psi+\sum_{j\neq i}\frac{\widehat\Omega_{ij}}{z_i-z_j}\psi,
    \end{align}
    where $A_i$ is a constant diagonal matrix specified by
    \begin{align}
        &A_i \phi^{(j,l)}=0,\\
        &A_i \phi^{(j,j)}=0,\\
        &A_i\phi^{(i,i)}=2\frac{\Lambda}{b^2}\phi^{(i,i)},\\
        &A_i\phi^{(i,j)}=\frac{\Lambda}{b^2}\phi^{(i,j)}.
    \end{align}
    This equation is equivalent to the height 2 $\widehat{sl(2)}$ irregular KZ equation if we consider the following rescaling analogous to the height 1 case:
    \begin{align}
        \psi\rightarrow \exp\left(-\sum_i\frac{\Lambda k_i z_i}{2b^2}\right)(z_i-z_j)^{\frac{-k_i k_j}{2b^2}}\psi.
    \end{align}
    
    \noindent\textbf{2)} The entries of $\psi$ are combinations of $\partial \phi$ and $\partial^2 \phi$, more specifically:
    \begin{align}
    & \frac{2}{b^4}\phi^{(i,i)}=-\frac{\partial^2 \phi}{\partial^2 z_i}+\sum_{j\neq i}\frac{1}{z_i-z_j} (\frac{k_j}{b^2} \frac{\partial \phi}{\partial z_i}-\frac{k_i}{b^2} \frac{\partial \phi}{\partial z_j})+\frac{\Lambda}{b^2}\frac{\partial \phi}{\partial z_i},\label{ii}\\
    &\frac{2}{b^4}\phi^{(i,j)}=-\frac{\partial^2 \phi}{\partial z_i\partial z_j}+\frac{1}{z_i-z_j}(-\frac{k_j}{b^2}\frac{\partial \phi}{\partial z_i}+\frac{k_i}{b^2}\frac{\partial \phi}{\partial z_j}).\label{ij} 
    \end{align}
\end{theorem}
The explicit proof is presented in appendix \ref{Proof1}, but as an example, we will prove the theorem in the case of the 2-point function, where the basis  is given by $v^{(1,1)},v^{(1,2)},v^{(2,2)}$.
The explicit form of $\Omega$ matrix under this basis becomes:
\begin{align}
    \Omega_{12}=\frac{k_1 k_2}{2}I_3+\begin{pmatrix}
-2k_2 & 2k_1-2 & 0 \\
 k_2& 2-(k_1+k_2) & k_1  \\
 0& 2k_2-2 & -2k_1 \\
\end{pmatrix}
\end{align}
We now show that for the 2-point function with 2 screening charges 
   \begin{align}
   \phi&=\int_{\Gamma} dw_1 dw_2 A(w_1,w_2,z),\\
    A(w_1,w_2,z)&=(w_1-w_2)^{-\frac{2}{b^2}} (w_1-z_1)^{\frac{k_1}{b^2}}(w_1-z_2)^{\frac{k_2}{b^2}}(w_2-z_1)^{\frac{k_1}{b^2}}(w_2-z_2)^{\frac{k_2}{b^2}},
\end{align}
the column vector
\begin{align}
    \psi=\begin{pmatrix}
\phi^{(1,1)} \\
\phi^{(1,2)} \\
\phi^{(2,2)}
\end{pmatrix}=\begin{pmatrix}
\int A(\frac{1}{(w_1-z_1)(w_2-z_1)})\\
\int A(\frac{1}{(w_1-z_2)(w_2-z_1)}+\frac{1}{(w_2-z_2)(w_1-z_1)}) \\
\int A(\frac{1}{(w_1-z_2)(w_2-z_2)}).
\end{pmatrix}
\end{align}
will satisfy the KZ equation:
\begin{align}\label{KZ2}
    -b^2 \partial_1 \psi=A_1\psi+\frac{\widehat{\Omega}_{12}^T}{z_1-z_2}\psi,   -b^2\partial_2 \psi=A_2\psi+\frac{\widehat{\Omega}_{21}^T}{z_2-z_1}\psi,
\end{align}
where
\begin{align}
\widehat{\Omega}_{12}=\widehat{\Omega}_{21}=\begin{pmatrix}
-2k_2 & 2k_1-2 & 0 \\
 k_2& 2-(k_1+k_2) & k_1  \\
 0& 2k_2-2 & -2k_1 \\
\end{pmatrix},
\end{align}
\begin{align}
    A_1=-\begin{pmatrix}
2\Lambda & 0 & 0 \\
 0& \Lambda & 0  \\
 0& 0 & 0 \\
\end{pmatrix},
A_2=-\begin{pmatrix}
0 & 0 & 0 \\
 0& \Lambda & 0  \\
 0& 0 & 2\Lambda \\
\end{pmatrix}.
\end{align}
\begin{proof}
    We first notice that 
\begin{align}
   \frac{\partial \phi^{(1,2)}}{\partial z_1\hfill}=\frac{\partial}{\partial z_1}(\int A(\frac{1}{(w_1-z_2)(w_2-z_1)}+\frac{1}{(w_2-z_2)(w_1-z_1)}))
\end{align}
And:
\begin{align}
    &\frac{\partial}{\partial z_1}\int A(\frac{1}{(w_1-z_2)(w_2-z_1)})\nonumber\\
    &=-\int A\frac{1}{w_1-z_2}((\frac{k_1}{b^2}-1)\frac{1}{(w_2-z_1)^2}+\frac{k_1}{b^2}\frac{1}{(w_2-z_1)(w_1-z_1)})\nonumber\\
    &=-\frac{k_1}{b^2}\frac{1}{z_1-z_2}\int A\frac{1}{(w_2-z_1)}(\frac{1}{(w_1-z_1)}-\frac{1}{(w_1-z_2)})-\int dw_1 \int \frac{A_{w_2,z_1}}{(w
_1-z_2)}d((w_2-z_1)^{\frac{k_1}{b^2}-1})\nonumber\\
&=...+\int A \frac{1}{(w_1-z_2)(w_2-z_1)}(\frac{2}{b^2}\frac{1}{w_1-w_2}+\frac{k_2}{b^2}\frac{1}{w_2-z_2}+\frac{\Lambda}{b^2})\nonumber\\
&=-\frac{k_1}{b^2}...+\frac{2}{b^2}...+\frac{1}{z_1-z_2}\frac{k_2}{b^2}\int A(\frac{1}{(w_1-z_2)}(\frac{1}{w_2-z_1}-\frac{1}{w_2-z_2}))+\frac{\Lambda}{b^2}\int A \frac{1}{(w_1-z_2)(w_2-z_1)}.
\end{align}
And similarly
\begin{align}
    &\frac{\partial}{\partial z_1}\int A(\frac{1}{(w_2-z_2)(w_1-z_1)})\nonumber\\
    &=-\frac{k_1}{b^2}\frac{1}{z_1-z_2}\int A\frac{1}{(w_1-z_1)}(\frac{1}{(w_2-z_1)}-\frac{1}{(w_2-z_2)})-\frac{2}{b^2}\int A \frac{1}{(w_2-z_2)(w_1-z_1)(w_1-w_2)}\nonumber\\
    &+\frac{1}{z_1-z_2}\frac{k_2}{b^2}\int A(\frac{1}{(w_2-z_2)}(\frac{1}{w_1-z_1}-\frac{1}{w_1-z_2}))+\frac{\Lambda}{b^2}\int A \frac{1}{(w_2-z_2)(w_1-z_1)}.
\end{align}
Notice that
\begin{align}
    &\int A (\frac{1}{(w_1-z_2)(w_2-z_1)(w_1-w_2)}-\frac{1}{(w_2-z_2)(w_1-z_1)(w_1-w_2)})\nonumber\\
    &=\int A \frac{1}{z_1-z_2}(\frac{1}{w_1-z_1}-\frac{1}{w_1-z_2})(\frac{1}{w_2-z_1}-\frac{1}{w_2-z_2})\nonumber\\
    &=\frac{1}{z_1-z_2}\int A(\frac{1}{(w_1-z_1)(w_2-z_1)}- (\frac{1}{(w_1-z_2)(w_2-z_1)}+\frac{1}{(w_2-z_2)(w_1-z_1)})+\frac{1}{(w_1-z_2)(w_2-z_2)}).
\end{align}
Combining them together, we get:
\begin{align}
    b^2\frac{\partial \phi^{(1,2)}}{\partial z_1\hfill}=\frac{\Lambda}{b^2}\phi^{(1,2)}+\frac{1}{z_1-z_2}((-2k_1+2)\phi^{(1,1)}+(k_1+k_2-2)\phi^{(1,2)}+(-2k_2+2)\phi^{(2,2)}).
\end{align}
And for $\phi^{(1,1)}$,we have
\begin{align}
    \frac{\partial \phi^{(1,1)}}{\partial z_1 \hfill}=-(\frac{k_1}{b^2}-1)\int A(\frac{1}{(w_1-z_1)^2(w_2-z_1)}+\frac{1}{(w_1-z_1)(w_2-z_1)^2}).
\end{align}
\begin{align}
    &\int A\frac{1}{(w_1-z_1)^2(w_2-z_1)}=\int dw_2 \int (\frac{k_1}{b^2}-1)^{-1} A_{w_1,z_1}\frac{1}{w_2-z_1}d((w_1-z_1)^{(\frac{k_1}{b^2}-1)})\nonumber\\
    &=-(\frac{k_1}{b^2}-1)^{-1}\int \frac{A}{(w_2-z_1)(w_1-z_1)}(-\frac{2}{b^2}\frac{1}{w_1-w_2}+\frac{k_2}{b^2}\frac{1}{w_1-z_2}+\frac{\Lambda}{b^2}),\nonumber
    \end{align}
And we notice that
\begin{align}
   &\int A(\frac{1}{(w_1-z_1)^2(w_2-z_1)}+\frac{1}{(w_1-z_1)^2(w_2-z_1)})\nonumber\\
   &=-(\frac{k_1}{b^2}-1)^{-1}\frac{2k_2}{b^2}\frac{1}{z_1-z_2}\int A \frac{1}{(w_1-z_1)(w_2-z_1)}\nonumber\\
   &+(\frac{k_1}{b^2}-1)^{-1}\frac{k_2}{b^2}\frac{1}{z_1-z_2}\int A (\frac{1}{(w_2-z_1)(w_1-z_2)}+\frac{1}{(w_1-z_1)(w_2-z_2)})+\frac{2\Lambda}{b^2}\int A \frac{1}{(w_1-z_1)(w_2-z_1)}.\nonumber
\end{align}
Thus
\begin{align}
    b^2\frac{\partial \phi^{(1,1)}}{\partial z_1\hfill}=2\frac{\Lambda}{b^2}\phi^{(1,1)}+\frac{1}{z_1-z_2}(2k_2\phi^{(1,1)}-k_2\phi^{(1,2)}).
\end{align}
And for $\phi^{(2,2)}$, we have
\begin{align}
   &\frac{\partial \phi^{(2,2)}}{\partial z_1\hfill}=-\frac{k_1}{b^2}\int A \frac{1}{(w_1-z_2)(w_2-z_2)}(\frac{1}{w_1-z_1}+\frac{1}{w_2-z_1})\nonumber\\
   &=-\frac{k_1}{b^2}\int A\left(\frac{1}{z_1-z_2}(\frac{1}{(w_2-z_2)(w_1-z_1)}-\frac{1}{(w_2-z_2)(w_1-z_2)})\right.\nonumber\\
   &\left.+\frac{1}{z_1-z_2}(\frac{1}{(w_1-z_2)(w_2-z_1)}-\frac{1}{(w_1-z_2)(w_2-z_2)})\right)\nonumber\\
   &=\frac{1}{b^2}\frac{1}{z_1-z_2}(-k_1\phi^{(1,2)}+2k_1\phi^{(2,2)}).
\end{align}
Thus the $\widehat{\Omega}$ matrix is of the following form:
\begin{align}
    -\begin{pmatrix}
2k_2 &-k_2  & 0 \\
-2k_1+2 & k_1+k_2-2 & -2k_2+2 \\
 0& -k_1 & 2k_1  \\
\end{pmatrix},
\end{align}
which is the conjugate of $\widehat\Omega_{12}$, and
\begin{align}
    A_1=\begin{pmatrix}
2\Lambda & 0 & 0 \\
 0& \Lambda & 0  \\
 0& 0 & 0 \\
\end{pmatrix}.
\end{align}
The other case involving $\partial_2$ can be shown similarly.

Furthermore, we can recover the original height 2 KZ equation after performing the following rescaling of the $\psi$ vector,
\begin{align}
    \psi\mapsto \left(
\begin{array}{ccc}
 (k_1-1) k_1& 0 & 0 \\
 0 & \frac{k_1 k_2}{2 } & 0 \\
 0 & 0 & (k_2-1) k_2 \\
\end{array}
\right) \psi,
\end{align}
which translates to
\begin{equation}
    \widehat{\Omega}^T \mapsto \widehat{\Omega}.
\end{equation}
\end{proof}

\subsection{Integral representations for arbitrary height}

Here we can generalize the above proofs to show that in more general cases irregular Liouville conformal blocks satisfy the KZ equation of arbitrary height.

\begin{theorem} \label{th:heightm}
     Define $N$-point degree 1 irregular Liouville CFT correlation function with $m$ screening operators:
    \begin{align}
        \mathcal{F}_{\Gamma}(z)=\exp\left(-\frac{\Lambda}{2b^2}\sum_{i=1}^N {k_i z_i}\right)\prod_{i<j}(z_i-z_j)^{-\frac{k_ik_j}{2b^2}}\int_{\Gamma} A({\bold{z,w}}) d^m w=\phi_0(z) \phi_{\Gamma}(z) ,
    \end{align}
    where
    \begin{align}
        A({\bold{z,w}})= \exp\left(\frac{\Lambda}{b^2}\sum_i w_i\right)\prod_{i<j}(w_i-w_j)^{-\frac{2}{b^2}}\prod_{i,j}(w_i-z_j)^\frac{k_j}{b^2}.
    \end{align}
   For integer $\bold{m}=(m_1,m_2,...,m_N)$, $m_i\geq 0$, $\sum_i m_i=m$, let
    \begin{align}
       \phi^{(\bold{m})}= \int A({\bold{z,w}}) d^m w \left(\prod_{i=1}^{m_1}\frac{1}{w_{i_1}-z_1}...\prod_{i=1}^{m_N}\frac{1}{w_{i_N}-z_N}\right).
    \end{align}
Then the column vector $\widehat\psi$ containing all $\phi^{(\bold{m})}$ will satisfy the following KZ equation
\begin{align}
   -b^2 \partial_i  \widehat\psi=A_i\widehat\psi+\sum_{j\neq i}\frac{\widehat{\Omega}^{T}_{ij}}{z_i-z_j}\widehat\psi,
\end{align}
where $\widehat{\Omega}^{T}$ is the transpose of the $\widehat{\Omega}$ matrix of height  $m$ KZ equation. And the action of $A_i$ is given by
\begin{align}
   A_i \phi^{(\bold{m})}=-\Lambda m_i  \phi^{(\bold{m})}.
\end{align}
Then the differential equation that the full conformal block  $\psi=\phi_0(z)\phi^{(\bold{m})}$ satisfies is the following:
\begin{align}\label{TmKZ}
    -b^2 \partial_i  \psi=\frac{\Lambda H_i}{2}\psi+\sum_{j\neq i}\frac{\Omega^{T}_{ij}}{z_i-z_j}\psi,
\end{align}
where 
\begin{align}\label{Hi}
    H_i \phi^{(\bold{m})}=(k_i-2m_i) \phi^{(\bold{m})}.
\end{align}

Equivalently, we can convert the above relations to height $m$ KZ equations. Let $v=v_1\otimes v_2\otimes...\otimes v_N$ be the highest weight vector as in \eqref{hiwei}. Define
\begin{align}
  F(z,w)=\sum_i T_i (z,w)=\exp(\frac{\Lambda w}{b^2})\prod_{j} (w-z_j)^\frac{k_j}{b^2} \sum_i \frac{Y_i}{w-z_i}=\phi_1(z,w)\sum_i \frac{Y_i}{w-z_i},
\end{align}
and\begin{align}
    \Psi&=\phi_0(z)\int \prod_{i\leq j} (w_i-w_j)^{\frac{-2}{b^2}}F(z,w_1)F(z,w_2)....F(z,w_N) v\, d^m \bold{w}\nonumber\\
    &=\phi_0(z)\int \rho(\bold{w})F(z,w_1)F(z,w_2)....F(z,w_m) v\, d^m \bold{w}\nonumber\\
    &=\phi_0(z) \int d^m \bold{w} \omega_{\bold{z,w}} v
\end{align}
Then \eqref{TmKZ} is equivalent to 
\begin{align}\label{mKZ}
   -b^2 \partial_i \Psi=\left(\frac{\Lambda}{2} H_i+\sum_{j\neq i}\frac{\Omega_{ij}}{z_i-z_j}\right)\Psi.
\end{align}
One can see that $\phi^{(\bold{m})}$ is the coefficient of the basis vector
    \begin{align*}
       Y_1^{m_1}Y_2^{m_2}...Y_N^{m_N}v
    \end{align*}
in the expansion of $\Psi$, which explains the transpose. Also, the $H_i$ matrix in \eqref{TmKZ} is the same as the action of the $H_i$ operator on $Y_1^{m_1}Y_2^{m_2}...Y_N^{m_N}v$, since by commutation relations we have
\begin{align}
    H_i Y_i^{m_i}v=(k_i-2m_i)Y_i^{m_i}v.
\end{align}
\end{theorem}
\begin{proof}
We now prove the operator version of the theorem. First, we notice that $\phi_0(z) v$ will satisfy \eqref{mKZ} from definition, since in this case $\Omega_{ij}$ acts on $v$ by scalar multiplication, and $H_i v=k_i v$. To prove \eqref{mKZ}, we first consider the following 
\begin{align}
    -b^2\partial_i F=F\frac{k_i}{w-z_i}+b^2\frac{\partial}{\partial w}T_i-\Lambda T_i-T_i\sum_j\frac{k_j}{w-z_j} ,
\end{align}
and
\begin{align}
    &\left[ \sum_{j\neq i}\frac{\Omega_{ij}}{z_i-z_j},F(z,w)\right]=F\frac{H_i}{w-z_i}-T_i\sum_j\frac{H_j}{w-z_j} ,\\
    &\left[ \frac{\Lambda}{2}H_i,F(z,w)\right]=-\Lambda T_i.
\end{align}
One gets 
\begin{align}\label{lem1}
      -b^2\partial_i F(z,w)-  \left[\frac{\Lambda}{2}H_i+\sum_{j\neq i}\frac{\Omega_{ij}}{z_i-z_j},F(z,w) \right]=F\frac{k_i-H_i}{w-z_i}-T_i\sum_{j}\frac{k_j-H_j}{w-z_j}+b^2\frac{\partial}{\partial w}T_i.
    \end{align}
In order to prove \eqref{mKZ}, we just need to prove that
\begin{align}
\int\left[D_i,\omega_{\bold{z,w}}\right]=\int\left[-b^2\partial_i-(\frac{\Lambda}{2}H_i+\sum_{j\neq i}\frac{\Omega_{ij}}{z_i-z_j}),\omega_{\bold{z,w}}\right]=0.
\end{align}
First, one has
\begin{align}
    \int\left[D_i,\omega_{\bold{z,w}}\right]&=\sum_{p=1}^m \int\rho(\bold{w})F(z,w_1)...\left[D_i,F(z,w_p)\right]...F(z,w_m)\nonumber\\
    &=\sum_{p=1}^m \int\rho(\bold{w})F(z,w_1)...\left(b^2\frac{\partial}{\partial w_p}T_i(z,w_p)-T_i(z,w_p)\sum_j\frac{k_j-H_j}{w_p-z_j}+F(z,w_p)\frac{k_i-H_i}{w_p-z_i}\right)...
\end{align}
For the first summand, one gets
\begin{align}
    &\sum_{p=1}^m \int\rho(\bold{w})F(z,w_1)...\left(b^2\frac{\partial}{\partial w_p}T_i(z,w_p)\right)...F(z,w_m) d^m \bold{w}\nonumber\\
    &=-\sum_{p=1}^m b^2\int\frac{\partial\rho(\bold{w})}{\partial w_p})F(z,w_1)...T_i(z,w_p)...F(z,w_m) d^m \bold{w}\nonumber\\
    &=\sum_{p,n:n\neq p}\int\frac{2\rho(\bold{w})}{w_p-w_n}(\prod_{l\neq n,p}F(w_l,z)) \phi_1(w_p,z)\phi_1(w_n,z)\left(\frac{Y_i}{(w_p-w_n)(w_p-z_i)}\right.\nonumber\\
    &\left.*\sum_j\frac{Y_j}{w_n-z_j} \right)v\,d^m \bold{w}
\end{align}
Similarly, by 
\begin{align}
    \left[k_i-H_i,Y(z,w)\right]= 2T_i(z,w),\left[\sum_j\frac{k_j-H_j}{w-z_j},Y(z,\tau)\right]= 2\sum_j\frac{T_j(z,\tau)}{w-z_j}
\end{align}
the second and third summands simplify to
\begin{align}
    &-2\sum_{n,p:n>p}\int \rho(\bold{w})F(z,w_1)...T_i(z,w_p)...(\sum_j \frac{T_j(z,w_n)}{w_p-w_j})...F(z,w_m)v\,d^m \bold{w}\nonumber\\
    &=-2\sum_{n,p:n>p}\int \rho(\bold{w})(\prod_{l\neq n,p}F(w_l,z)) \phi_1(w_p,z)\phi_1(w_n,z) \frac{Y_i}{w_p-z_i}\sum_j \frac{Y_j}{(w_n-z_j)(w_p-z_j)}v\,d^m \bold{w}
\end{align}
and
\begin{align}
    &2\sum_{n,p:n>p}\int \rho(\bold{w})F(z,w_1)...( \frac{T_i(z,w_n)}{w_p-z_i})...F(z,w_m)v\,d^m \bold{w}\nonumber\\
    &=2\sum_{n,p:n>p}\int \rho(\bold{w})(\prod_{l\neq n,p}F(w_l,z)) \phi_1(w_p,z)\phi_1(w_n,z) \frac{Y_i}{(w_p-z_i)(w_n-z_i)}\sum_j \frac{Y_j}{(w_p-z_j)}v\,d^m \bold{w}
\end{align}
And we have 
\begin{align}
    &\sum_{p,n:n\neq p}\frac{Y_i}{(w_p-w_n)(w_p-z_i)}\sum_j \frac{Y_j}{(w_n-z_j)}-\sum_{n,p:n>p}\frac{Y_i}{w_p-z_i}\sum_j \frac{Y_j}{(w_n-z_j)(w_p-z_j)}\nonumber\\
    &+\sum_{n,p:n>p}\frac{Y_i}{(w_p-z_i)(w_n-z_i)}\sum_j \frac{Y_j}{(w_p-z_j)},
\end{align}
which vanishes since
\begin{align}
    &\sum_{p,n:n\neq p}\frac{Y_i}{(w_p-w_n)(w_p-z_i)}\sum_j \frac{Y_j}{(w_n-z_j)}=\sum_{p,n:n>p}\frac{Y_i}{(w_p-w_n)(w_p-z_i)}\sum_j\frac{Y_j}{w_n-z_j}\nonumber\\
    &-\sum_{p,n:n>p}\frac{Y_i}{(w_p-w_n)(w_n-z_i)}\sum_j\frac{Y_j}{w_p-z_j}.
\end{align}
Thus we proved the irregular version of the KZ equations.
\end{proof}

\section{Flatness and the space at infinity}
In this section, we want to analyze conditions under which irregular KZ equations of arbitrary degree satisfy the flatness condition. To this end, we will start with an irregular operator at finite distance and subsequently send the operator to infinity. 

\subsection{Degree $1$}

The KZ equation  \eqref{height1kz} defines a flat connection. To see this, we define
\begin{equation}
    \Omega_i \equiv \sum_{j \neq i} \frac{\Omega_{ij}}{z_i - z_j},
\end{equation}
and note that 
\begin{equation} \label{eq:OmegaHcom}
    [\Omega_i, H_i] = \left[\sum_{j \neq i} \frac{\Omega_{ij}}{z_i-z_j}, H_i\right] = \sum_{j \neq i} \frac{2}{z_i - z_j} (Y_i X_j - X_i Y_j).
\end{equation}
Then, for the connection
\begin{equation}
    \nabla_i \equiv \frac{\partial}{\partial z_i} - \frac{\Lambda}{2} H_i - \sum_{j \neq i} \frac{\Omega_{ij}}{z_i -z_j}, 
\end{equation}
we compute the commutator
\begin{align}
    [\nabla_i, \nabla_j] &= \left[\partial_i - \frac{\Lambda}{2} H_i - \Omega_i, \partial_j - \frac{\Lambda}{2} H_j - \Omega_j\right] \nonumber \\
    ~ &= \left[\partial_i -\Omega_i, \partial_j - \Omega_j\right] + \left[\frac{\Lambda}{2}H_i, \Omega_j\right] - \left[\frac{\Lambda}{2} H_j, \Omega_i\right] \nonumber \\
    &= \sum_{j \neq i} \frac{\Lambda}{z_i - z_j}(Y_i X_j - X_i Yj + Y_j X_i - X_j Y_i) \nonumber \\
    &= 0,
\end{align}
by using \eqref{eq:OmegaHcom} and the fact that the connection is flat for $\Lambda = 0$.

\subsection{Degree $r > 1$}

The degree $2$ irregular KZ equation given in equation \eqref{eq:r2KZ} can be extended to arbitrary degree by using degree $r$ irregular Kac-Moody modules at infinity. The general form of such an equation will be as follows:
\begin{equation} \label{eq:RankriKZi}
    \kappa \frac{\partial \psi}{\partial z_i} = \left(\sum_{l=0}^{r-1} A_i^{(l)} z_i^l + \sum_{j \neq i} \frac{\Omega_{ij}}{z_i - z_j}\right).
\end{equation}

In order to see under what conditions this defines a flat connection, we follow the approach of \cite{Resh-KZ} and start with an irregular singularity at the finite distance point $z_1$. The corresponding covariant derivative is given by
\begin{equation} \label{eq:RankriKZ}
    \nabla_i =  \frac{\partial}{\partial z_i} - \sum_{j \neq i,1} \frac{\Omega_{ij}}{z_i - z_j} - \sum_{l=1}^r \frac{\Omega_{i1}^{(l)}}{(z_i - z_1)^{l+1}},
\end{equation}
where $\Omega^{(l)} \in \mathfrak{g}^{(l)} \otimes \mathfrak{g}$ with 
\begin{equation}
    [ \mathfrak{g}^{(l)}, \mathfrak{g}^{(k)}] \in \mathfrak{g}^{(l+k)}.
\end{equation}
$\mathfrak{g}^{(l)}$ is the degree $l$ subspace of the affine Lie algebra $\widehat{\mathfrak{g}}$ and elements of it are sometimes denotes as $X t^l$ with $X \in \mathfrak{g}$. One can easily check that the connection \eqref{eq:RankriKZ} satisfies the flatness condition if we impose
\begin{equation}
    [\Omega_{ij}, \Omega_{1i}^{(l)} + \Omega_{1j}^{(l)}] = 0.
\end{equation}
Next, we want to perform a change of variables
\begin{equation}
    w_i \equiv \frac{1}{z_i-z_1},
\end{equation}
giving
\begin{equation}
    \nabla_i = - w_i^2 \frac{\partial}{\partial w_i} - \sum_{j \neq i, 1} \frac{\Omega_{ij}}{w_i^{-1} - w_j^{-1}} - \sum_{l=1}^r \Omega_{i1}^{(l)} w_i^{l+1}.
\end{equation}
Sending $z_1 \rightarrow \infty$, we get (using $\frac{w_j}{w_i} \rightarrow 1$),
\begin{equation}
    -\frac{1}{w_i^2} \nabla_i = \frac{\partial}{\partial w_i} - \sum_{j \neq i,1} \frac{\Omega_{ij}}{w_i - w_j} - \sum_{l=1}^r \Omega_{i1}^{(l)} w_i^{l-1}.
\end{equation}
By abuse of notation, we see that 
\begin{equation}
    \nabla_i = \frac{\partial}{\partial z_i} - \sum_{j \neq i} \frac{\Omega_{ij}}{z_i - z_j} - \sum_{l=1}^r \Omega_{\infty i}^{(l)} z_i^{l-1},
\end{equation}
defines a flat connection. Comparing to \eqref{eq:RankriKZi} and identifying $A_i^{(l)} = \Omega_{\infty i}^{(l+1)}$, we see that we can satisfy flatness if we allow $A_i^{(l)} \in \mathfrak{g}^{(l)} \otimes \mathfrak{g}$ such that $\mathfrak{g}^{(l)}$ acts at infinity. In particular, all currents $\alpha(z)$, $\beta(z)$ and $\gamma(z)$ must act irregularly on the module at infinity. Thus the highest weight vacuum vector has to be extended to include the representation at infinity as follows
\begin{equation}
    v_{\infty} \otimes v_1 \otimes \cdots \otimes v_n,
\end{equation}
where $v_{\infty}$ is now the vector 
\begin{equation}
    v_{\infty} \equiv \bigoplus_{l=0}^{r-1} v_{\infty}^{(l)}.
\end{equation}
This has as a consequence that now height one representations are $n+r$ dimensional due to the action of $Y^{(l)}$ at infinity.

\section{Applications}

\subsection{Surface operators in Argyres-Douglas theories}

The space of conformal blocks in Liouville theory can be identified with the Hilbert space of the corresponding 4d theory on $\mathbb{R} \times S^3$ in the Omega-background \cite{Nekrasov:2010ka}. In the present paper we considered the space of conformal blocks of $N$ degenerate fields on $\mathbb{C}$ with one irregular operator at infinity. Equivalently, via a projective compactification, this can be described as Liouville theory on $C = \mathbb{P}^1$ with $N$ degenerate fields and one irregular operator. In four-dimensional theory, the irregular operator is part of the definition of the theory, in partcular, our present setup describes an Argyres-Douglas theory \cite{Argyres:1995jj,Gaiotto:2009ma,Gaiotto:2012sf}, whereas degenerate fields correspond to surface operators supported on two-dimensional subspace of the 4d space-time \cite{Alday:2009fs}. Specifically, in our setup the surface operators are supported on $\mathbb{R} \times S^1$, where $S^1 \subset S^3$ is the great circle.

Therefore, monodromies of the KZ flat connection studied in this paper describe braiding of surface operators in the $\mathbb{R} \times C$ part of the space-time, where they look like lines. A convenient way to see this is via a six-dimensional perspective of the $(0,2)$ fivebrane theory on $C \times \mathbb{R} \times S^3$. In the six-dimensional setup the degenerate fields correspond to codimension-4 defects whereas the irregular operator is represented by a codimension-2 defect supported on all of $\mathbb{R} \times S^3$. The function $\mathcal{W}$ in \eqref{eq:Wirr} is the twisted superpotential of the 2d effective theory on surface operators. This perspective is consistent with 3d-3d correspondence that involves partial topological twist along the $\mathbb{R} \times C$ part of the space-time. One can also consider conformal blocks in ``half of Liouville theory'' that breaks the symmetry between $b^2$ and $b^{-2}$; this can be achieved by cutting $S^3$ into two copies of $S^1 \times D^2$ glued along a 2-torus $T^2$. Note compactification on this 2-torus results in 4d $\mathcal{N}=4$ super-Yang-Mills on $\mathbb{R} \times I \times C$ which, when further compactified on $C$, can be described by 2d sigma-model \cite{Harvey:1995tg,Bershadsky:1995vm} on a strip $\mathbb{R} \times I$ with target space $\mathcal{M}_H (SU(2), C)$, the space of Higgs bundles on $C$ with tame and wild ramification.

Also note that, if instead of $N$ degenerate fields we considered $N$ copies of generic Liouville operators (with momenta $\alpha_i$), then instead of braiding of surface operators the monodromy of the KZ connection would correspond to a domain wall in 4d theory, namely a duality wall. More generally, the mapping class group of $C$ acts as a duality group on 4d theory and monodromies associated with braiding codimension-2 defects represent duality walls. Below we offer another perspective on these monodromies.

\subsection{Rozansky-Witten theories}

Once we are discussing a higher-dimensional QFT perspective on the KZ monodromy involving irregular operators, it is instructive to make a comparison with the category $\text{MTC} [C \times S^1]$ of line operators --- in particular, their braiding --- in a Rozansky-Witten theory based on a hyper-Kahler target space $\mathcal{M}_H (SU(2), C)$. The latter is obtained by a topological twist \cite{Rozansky:1996bq} of 3d $\mathcal{N}=4$ theory obtained via compactification of 6d $(0,2)$ theory of type $A_1$ on $C \times S^1$.

The above mentioned duality relations allow us to identify $\text{MTC} [C \times S^1]$ with the (suitable version of) the derived category of coherent sheaves $D^b (\mathcal{M}_H (SU(2), C))$ \cite{Dedushenko:2018bpp,Gukov:2024adb}. Note, the rank of the Grothendieck group of $D^b (\mathcal{M}_H (SU(2), C))$ grows as $2^N$, where $N$ is the number of tame ramification points. This agrees with the dimension of the space of conformal blocks in section \ref{sec:heights}. Indeed, codimension-4 defects in 6d that represent degenerate fields on $C$ correspond to line operators of the Rozansky-Witten theory. One crucial difference is that the setup relevant to Liouville theory involves Omega-background, whereas the one relevant to the Rozansky-Witten theory does not. As a result, global monodromies in the latter system, classified by the group $\text{Auteq} \, D^b (\mathcal{M}_H (SU(2), C))$, are independent of the parameter $b$. Therefore, the global monodromies in the Liouville theory with irregular and degenerate fields can be considered as deformations of the global monodromies in this Rozansky-Witten setup.

We can also justify the relation between the two systems via 4d $\mathcal{N} = 4$ super-Yang-Mills obtained by first compactifying on the 2-torus $T^2$. In both cases, one finds 4d theory on $\mathbb{R} \times I \times C$, and the only difference between the two systems is the choice of the boundary conditions $\mathcal{B}_1$ and $\mathcal{B}_2$ on the two sides of the strip $\mathbb{R} \times I$. To the extent the two sets of boundary conditions can be considered deformations of one another, cf. \cite{Nekrasov:2010ka}, the same can be said about the global monodromies in the two systems. In both cases, further reduction on $C$ yields a sigma-model on $\mathbb{R} \times I$ with the target space $\mathcal{M}_H (SU(2), C)$.

\acknowledgments{It is our pleasure to thank Boris Feigin, Xia Gu, Pavel Putrov and Xiaomeng Xu for helpful discussions and suggestions.
The work of S.G. is supported by a Simons Collaboration Grant on New Structures in Low-Dimensional Topology, by the NSF grant DMS-2245099, and by the U.S. Department of Energy, Office of Science, Office of High Energy Physics, under Award No. DE-SC0011632.} The work of B.H. and Y.L. is supported by NSFC grant 12250610187. The work of N.R. was supported by the Simons Collaboration ``Categorical symmetries'',  by the grant BMSTC and ACZSP (Grant no. Z221100002722017), by the Changjiang fund, and by the project No 075-15-2024-631 funded by the ministery of Science and Higher Education of the Russian Federation.

\appendix

\section{Proof of Theorem \ref{th:height2}}\label{Proof1}
Here we give a proof of theorem \ref{th:height2}:
\begin{proof}
To prove (1), we just need to consider the following(we have the convention that $i \neq j,i\neq l, j\neq l$) :
\begin{align}
    \partial_l \phi^{(i,j)}&=\frac{-k_l}{b^2}\frac{k_i k_j}{2}\int A (\frac{1}{(w_1-z_i)(w_2-z_j)}+\frac{1}{(w_2-z_i)(w_1-z_j)})(\frac{1}{w_1-z_l}+\frac{1}{w_2-z_l})\nonumber\\
    &=\frac{-k_l}{b^2}\frac{k_i k_j}{2}\int A\left(\frac{1}{z_l-z_i}(\frac{1}{(w_1-z_l)(w_2-z_j)}-\frac{1}{(w_1-z_i)(w_2-z_j)})\right.\nonumber\\
    &+\frac{1}{z_l-z_j}(\frac{1}{(w_1-z_l)(w_2-z_i)}-\frac{1}{(w_1-z_j)(w_2-z_i)})\nonumber\\
    &+\frac{1}{z_l-z_j}(\frac{1}{(w_2-z_l)(w_1-z_i)}-\frac{1}{(w_2-z_j)(w_1-z_i)})\nonumber\\
    &+\left.\frac{1}{z_l-z_i}(\frac{1}{(w_2-z_l)(w_1-z_j)}-\frac{1}{(w_2-z_i)(w_1-z_j)}\right)\nonumber\\
    &=\frac{1}{b^2}\frac{1}{z_l-z_i}(k_l\phi^{(i,j)}-k_i\phi^{(l,j)})+\frac{1}{b^2}\frac{1}{z_l-z_j}(k_l\phi^{(i,j)}-k_j\phi^{(l,i)})\label{lij}.
    \end{align}
    \ref{lij} is equivalent to part of \ref{abad} and \ref{abcd} when $c\neq d$.And we also have
 \begin{align}
        \partial_j \phi^{(i,i)}&=-\frac{k_j}{b^2}k_i(k_i-1)\int A \frac{1}{(w_1-z_i)(w_2-z_i)}(\frac{1}{w_1-z_j}+\frac{1}{w_2-z_j})\nonumber\\
       &=\frac{1}{b^2}\frac{1}{z_j-z_i}(2k_j\phi^{(i,i)}-2(k_i-1)\phi^{(i,j)}). \label{jii}
 \end{align}
\ref{jii} is equivalent to part of \ref{abaa} and  \ref{abcd} when $c=d$.And

\begin{align}
    \partial_i \phi^{(i,i)}&=-k_i(k_i-1)(\frac{k_1}{b^2}-1)\int A(\frac{1}{(w_1-z_i)^2(w_2-z_i)}+\frac{1}{(w_1-z_i)(w_2-z_i)^2}),\nonumber\\
    \int A\frac{1}{(w_1-z_i)^2(w_2-z_i)}&=\int dw_2 \int (\frac{k_i}{b^2}-1)^{-1} A_{w_1,z_i}\frac{1}{w_2-z_i}d((w_1-z_i)^{(\frac{k_i}{b^2}-1)})\nonumber\\
    &=-(\frac{k_i}{b^2}-1)^{-1}\int \frac{A}{(w_2-z_i)(w_1-z_i)}(-\frac{2}{b^2}\frac{1}{w_1-w_2}+\sum_{j\neq i}\frac{k_j}{b^2}\frac{1}{w_1-z_j}),\nonumber
    \end{align}
\begin{align}
    &\int A(\frac{1}{(w_1-z_i)^2(w_2-z_i)}+\frac{1}{(w_2-z_1)^2(w_1-z_i)})\nonumber\\
    &=-(\frac{k_i}{b^2}-1)^{-1}\int A\sum_{j\neq i}\frac{k_j}{b^2}(\frac{1}{(w_2-z_i)(w_1-z_i)}(\frac{1}{w_1-z_j}+\frac{1}{w_2-z_j}),
\end{align}
    thus
\begin{align}
   \partial_i \phi^{(i,i)}&=k_i(k_i-1)\int A\sum_{j\neq i}\frac{k_j}{b^2}(\frac{1}{(w_2-z_i)(w_1-z_i)}(\frac{1}{w_1-z_j}+\frac{1}{w_2-z_j})\nonumber\\
   &=\frac{1}{b^2}\sum_{j\neq i}\frac{1}{z_i-z_j}(2k_j\phi^{(i,i)}-2(k_i-1)\phi^{(i,j)}). \label{iii}
\end{align}
\ref{iii} with \ref{jii} together is equivalent to \ref{abaa}. Finally
\begin{align}
\partial_i \phi^{(i,j)}=\frac{k_i k_j}{2}\partial_i(\int A(\frac{1}{(w_1-z_j)(w_2-z_i)}+\frac{1}{(w_2-z_j)(w_1-z_i)}))
\end{align}

And:
\begin{align}
    &\frac{\partial}{\partial z_i}\int A(\frac{1}{(w_1-z_j)(w_2-z_i)})\nonumber\\
    &=-\int A\frac{1}{w_1-z_j}((\frac{k_i}{b^2}-1)\frac{1}{(w_2-z_i)^2}+\frac{k_i}{b^2}\frac{1}{(w_2-z_i)(w_1-z_i)})\nonumber\\
    &=-(\frac{k_i}{b^2}-1)\int A\frac{1}{(w_1-z_j)(w_2-z_i)^2}-\frac{k_i}{b^2}\int A\frac{1}{(w_1-z_j)(w_2-z_i)(w_1-z_i)}\nonumber\\
    &=-(\frac{k_i}{b^2}-1)\int A\frac{1}{(w_1-z_j)(w_2-z_i)^2}-\frac{k_i}{b^2}\frac{1}{z_i-z_j}\int A\frac{1}{(w_2-z_i)}(\frac{1}{(w_1-z_i)}-\frac{1}{(w_1-z_j)})\nonumber\\
    &=-\frac{k_i}{b^2}\frac{1}{z_i-z_j}\int A\frac{1}{(w_2-z_i)}(\frac{1}{(w_1-z_i)}-\frac{1}{(w_1-z_j)})-\int dw_1 \int \frac{A_{w_2,z_i}}{(w
_1-z_j)}d((w_2-z_i)^{\frac{k_i}{b^2}-1})\nonumber\\
&=...+\int A \frac{1}{(w_1-z_j)(w_2-z_i)}(\frac{2}{b^2}\frac{1}{w_1-w_2}+\sum_{l\neq i}\frac{k_l}{b^2}\frac{1}{w_2-z_l})\nonumber\\
&=...+\frac{2}{b^2}\int A \frac{1}{(w_1-z_j)(w_2-z_i)(w_1-w_2)}+\sum_{l\neq i}\frac{k_l}{b^2}\int A \frac{1}{(w_1-z_j)(w_2-z_i)(w_2-z_l)}\nonumber\\
&=-\frac{k_i}{b^2}...+\frac{2}{b^2}...+\sum_{l\neq i}\frac{k_l}{b^2}\frac{1}{z_i-z_l}\int A(\frac{1}{(w_1-z_j)}(\frac{1}{w_2-z_i}-\frac{1}{w_2-z_l})).
\end{align}
And similarly
\begin{align}
    &\frac{\partial}{\partial z_i}\int A(\frac{1}{(w_2-z_j)(w_1-z_i)})\nonumber\\
    &=-\frac{k_i}{b^2}\frac{1}{z_i-z_j}\int A\frac{1}{(w_1-z_i)}(\frac{1}{(w_2-z_i)}-\frac{1}{(w_2-z_j)})-\frac{2}{b^2}\int A \frac{1}{(w_2-z_j)(w_1-z_i)(w_1-w_2)}\nonumber\\
    &+\sum_{l\neq i}\frac{1}{z_i-z_l}\frac{k_l}{b^2}\int A(\frac{1}{(w_2-z_j)}(\frac{1}{w_1-z_i}-\frac{1}{w_1-z_l})).
\end{align}
Notice that
\begin{align}
    &\int A (\frac{1}{(w_1-z_j)(w_2-z_i)(w_1-w_2)}-\frac{1}{(w_2-z_j)(w_1-z_i)(w_1-w_2)})\nonumber\\
    &=\int A \frac{1}{z_i-z_j}(\frac{1}{w_1-z_i}-\frac{1}{w_1-z_j})(\frac{1}{w_2-z_i}-\frac{1}{w_2-z_j})\nonumber\\
    &=\frac{1}{z_i-z_j}\int A(\frac{1}{(w_1-z_i)(w_2-z_i)}- (\frac{1}{(w_1-z_j)(w_2-z_i)}+\frac{1}{(w_2-z_j)(w_1-z_i)})+\frac{1}{(w_1-z_j)(w_2-z_j)})
\end{align}
Combining them together, we get:
\begin{align}
    b^2\partial_i \phi^{(i,j)}&=\frac{k_ik
_j}{2}\frac{1}{z_i-z_j}((-2k_i+2)\frac{1}{k_i(k_i-1)}\phi^{(i,i)}+(k_i+k_j-2)\frac{2}{k_ik_j}\phi^{(i,j)}+(-2k_j+2)\frac{1}{k_j(k_j-1)}\phi^{(j,j)})\nonumber\\
&+\frac{k_ik
_j}{2}\sum_{l\neq {i,j}}\frac{1}{z_i-z_l}k_l(\frac{2}{k_i k_j}\phi^{(i,j)}-\frac{2}{k_j k_l}\phi^{(j,l)})\nonumber\\
&=\frac{1}{z_i-z_j}(-k_j\phi^{(i,i)}+(k_i+k_j-2)\phi^{(i,j)}+k_i\phi^{(j,j)})+\sum_{l\neq {i,j}}\frac{1}{z_i-z_l}(k_l\phi^{(i,j)}-k_i\phi^{(j,l)}).\label{iij}
\end{align}
\ref{iij} together with \ref{lij} is equivalent to \ref{abab} and \ref{abad}. Thus (1) is proven.

Now we prove (2). For \ref{ii}, we have:
\begin{align}
    \frac{\partial^2 \phi}{\partial^2 z_i}=\int  A\left(\frac{k_i}{b^2}(\frac{k_i}{b^2}-1)\frac{1}{(w_1-z_i)^2}+2\frac{k_i}{b^2}\frac{k_i}{b^2}\frac{1}{(w_1-z_i)(w_2-z_i)}+\frac{k_i}{b^2}(\frac{k_i}{b^2}-1)\frac{1}{(w_2-z_i)^2} \right)
\end{align}
And 
\begin{align}
    &\int A\left(\frac{k_i}{b^2}(\frac{k_i}{b^2}-1)\frac{1}{(w_1-z_i)^2}\right)\nonumber\\
    &=\int dw_2 \int A_{w_1,z_i}*(\frac{k_i}{b^2}d((w_1-z_i)^{\left(\frac{k_i}{b^2}-1\right)}),\text{where $A_{w_1,z_i}=A*\frac{1}{(w_1-z_i)^{\frac{k_i}{b^2}}}$}\nonumber\\
    &=-\frac{k_i}{b^2}\int dw_2 d(A_{w_1,z_i})(w_1-z_i)^{\left(\frac{k_i}{b^2}-1\right)}\nonumber\\
    &=-\frac{k_i}{b^2}\int  A (w_1-z_i)^{-1} (\frac{-2}{b^2}\frac{1}{(w_1-w_2)}+\sum_{j\neq i}\frac{k_i}{b^2}\frac{1}{(w_1-z_j)}+\frac{\Lambda}{b^2})\nonumber\\
    &=-\frac{k_i}{b^2}\int A\left(\frac{-2}{b^2}\frac{1}{(w_1-z_i)(w_1-w_2)}+\sum_{j\neq i}\frac{k_j}{b^2}\frac{1}{(w_1-z_i)(w_1-z_j)}+\frac{\Lambda}{b^2}\frac{1}{w_1-z_i}\right)\nonumber\\
    &=-\frac{k_i}{b^2}\int  A(\frac{-2}{b^2}\frac{1}{(w_1-w_2)(w_1-z_i)} )\nonumber\\
    &-\frac{k_i}{b^2}\int  A \sum_{j\neq i}\frac{k_j}{b^2}(\frac{1}{z_i-z_j}(\frac{1}{w_1-z_i}-\frac{1}{w_1-z_j}))-\frac{\Lambda
    }{b^2}\frac{k_i}{b^2}\int A \frac{1}{w-z_i}.
\end{align}
Similarly, we have
\begin{align}
     &\int A\left(\frac{k_i}{b^2}(\frac{k_i}{b^2}-1)\frac{1}{(w_2-z_i)^2}\right)\nonumber\\
     &=-\frac{k_i}{b^2}\int  A(\frac{2}{b^2}\frac{1}{(w_1-w_2)(w_2-z_i)} )\nonumber\\
    &-\frac{k_i}{b^2}\int A\sum_{j\neq i}\frac{k_j}{b^2}(\frac{1}{z_i-z_j}(\frac{1}{w_2-z_i}-\frac{1}{w_2-z_j}))-\frac{\Lambda}{b^2}\frac{k_i}{b^2}\int A \frac{1}{w_2-z_i}.
\end{align}
Thus
\begin{align}
    \frac{\partial^2 \phi}{\partial^2 z_i}&=-\frac{2}{b^4}\phi^{(i,i)}+\sum_{j\neq i}\frac{1}{z_i-z_j} (\frac{k_j}{b^2} \frac{\partial \phi}{\partial z_i}-\frac{k_i}{b^2} \frac{\partial \phi}{\partial z_j})+\frac{\Lambda}{b^2}\frac{\partial \phi}{\partial z_i}.
\end{align}
As for \ref{ij}, we have:
\begin{align}
    \frac{\partial^2 \phi}{\partial z_i\partial z_j}&=\frac{k_i k_j}{b
    ^4}\int A (\frac{1}{w_1-z_i}+\frac{1}{w_2-z_i}) (\frac{1}{w_1-z_j}+\frac{1}{w_2-z_j})\nonumber\\
    &=\frac{2}{b^4}\phi^{(i,j)}+\frac{k_ik_j}{b^4}\frac{1}{z_i-z_j}(\frac{1}{w_1-z_i}-\frac{1}{w_1-z_j}+\frac{1}{w_2-z_i}-\frac{1}{w_2-z_j})\nonumber\\
    &=\frac{2}{b^4}\phi^{(i,j)}+\frac{1}{z_i-z_j}(-\frac{k_j}{b^2}\frac{\partial \phi}{\partial z_i}+\frac{k_i}{b^2}\frac{\partial \phi}{\partial z_j})
\end{align}

which completes the theorem.
\end{proof}
\newpage
\bibliographystyle{JHEP}     
 {\small{\bibliography{main}}}

\end{document}